%% file: ms.tex
\newif\iffull
\fulltrue  %full version
\newif\ifieeesp %oakland version
\ifieeesp
\documentclass[conference]{IEEEtran}

\def\BibTeX{{\rm B\kern-.05em{\sc i\kern-.025em b}\kern-.08em
    T\kern-.1667em\lower.7ex\hbox{E}\kern-.125emX}}
\fi

\iffull
\documentclass[a4paper,english,11pt]{article}

\usepackage{fullpage}
\usepackage{amsthm}
\newcommand*{\affaddr}[1]{#1} % No op here. Customize it for different styles.
\newcommand*{\affmark}[1][*]{\textsuperscript{#1}}

\usepackage[table, svgnames,dvipsnames]{xcolor}
\usepackage{wrapfig}
\fi

\usepackage[english]{babel}
\usepackage{soul}
\ifieeesp
\usepackage[table, svgnames, dvipsnames]{xcolor}
\fi
\usepackage{booktabs}
\usepackage[utf8]{inputenc}
\usepackage{color}
\usepackage{eurosym}
\definecolor{dkgreen}{rgb}{0,0.6,0}
\definecolor{gray}{rgb}{0.5,0.5,0.5}
\definecolor{mauve}{rgb}{0.58,0,0.82}
\usepackage{listings}
\usepackage{capt-of}

\usepackage[users=BAD]{uchanges}
\ChangesSetUserColor[Bernardo]{B}{red}
\ChangesSetUserColor[Aravinda]{A}{blue}
\ChangesSetUserColor[Dominic]{D}{orange}

\renewcommand{\paragraph}[1]{\vspace{1mm}\noindent \textbf{#1.\ }}

\usepackage{caption, subcaption}
\usepackage{amsfonts,amssymb,amstext,enumerate}
\usepackage{array}
\usepackage[algoruled,linesnumbered,vlined,noend]{algorithm2e}
\SetNlSty{text}{}{:}
	\SetKwComment{Comment}{$\triangleright$\ }{}

\usepackage{graphicx}
\usepackage{listing}
\usepackage{mathtools}
\usepackage{adjustbox}
\usepackage{float}
\usepackage{tikz}
\usetikzlibrary{arrows.meta,decorations.pathmorphing,decorations.footprints,fadings,calc,trees,mindmap,shadows,decorations.text,patterns,positioning,shapes,matrix,fit}
\usepackage{pgfplots}
\usepackage{enumitem}
\setlist{noitemsep}
\usepackage{xspace}
\usepackage{amsthm}
\usepackage{float}
\floatstyle{boxed}
\restylefloat{figure}

\iffull \usepackage{hyperref} \fi
\usepackage{bm}
\iffull \hypersetup{
    colorlinks = true,
    citebordercolor = {lime},
    linkbordercolor = {purple},
    urlbordercolor  = {magenta},
    citecolor = {teal},
    linkcolor = {purple},
    urlcolor  = {magenta}
}\fi
\definecolor{light-gray}{gray}{0.80}
\usepackage[capitalize]{cleveref}
\usepackage{tikz}
\usepackage{anyfontsize}
\usetikzlibrary{decorations.pathreplacing}

\usetikzlibrary{arrows.meta,decorations.pathmorphing,decorations.footprints,fadings,calc,trees,mindmap,shadows,decorations.text,patterns,positioning,shapes,matrix,fit,shapes.misc}
\usetikzlibrary{shapes.geometric, arrows, fit}

\tikzstyle{block} = [rectangle, minimum width=1cm, minimum height=1cm,text centered, draw=black]
\tikzstyle{arrow} = [thick,->,>=stealth]

\ifieeesp
\title{Redactable Blockchain in the Permissionless Setting}
\author{\IEEEauthorblockN{Dominic Deuber\IEEEauthorrefmark{1}, Bernardo Magri\IEEEauthorrefmark{2} and Sri Aravinda Krishnan Thyagarajan\IEEEauthorrefmark{1}} \IEEEauthorblockA{\IEEEauthorrefmark{1}Friedrich-Alexander University Erlangen-N\"urnberg, Germany\\
E-mail: \{deuber, thyagarajan\}@cs.fau.de} \IEEEauthorblockA{\IEEEauthorrefmark{2}
Aarhus University, Denmark\\
E-mail: magri@cs.au.dk}
}

\fi

\usepackage{tabularx}

\tikzset{%
    box/.style    = {
rectangle split, rectangle split, 
                      rectangle split parts={#1},
                       font            = \sffamily\footnotesize,
                       text width      = 4cm, 
                       draw
                        },
  wbox/.style    = {
rectangle, 
                       font            = \sffamily\footnotesize,
                       text width      = 4cm, 
                       yshift=2.5em,
                       draw,
                       dashed
                        },
 }

\makeatletter
\DeclareRobustCommand{\rvdots}{%
  \vbox{
	\vspace{0.7em}
    \baselineskip4\p@\lineskiplimit\z@
    \kern-\p@
    \hbox{.}\hbox{.}\hbox{.}
    \vspace{0.7em}
  }}
\makeatother

\makeatletter
\def\thickhline{%
  \noalign{\ifnum0=`}\fi\hrule \@height \thickarrayrulewidth \futurelet
   \reserved@a\@xthickhline}
\def\@xthickhline{\ifx\reserved@a\thickhline
               \vskip\doublerulesep
               \vskip-\thickarrayrulewidth
             \fi
      \ifnum0=`{\fi}}
\makeatother

\newlength{\thickarrayrulewidth}
\input{macro}

\iffull
\title{Redactable Blockchain in the Permissionless Setting}
\author{%
Dominic Deuber\affmark[1], Bernardo Magri\affmark[2]\thanks{Work done while the author was affiliated with Friedrich-Alexander University Erlangen-N\"urnberg.}\  and Sri Aravinda Krishnan Thyagarajan\affmark[1]\\ \\
\affaddr{\affmark[1]\textit{Friedrich-Alexander-Universit\"at Erlangen-N\"urnberg, Germany}}\\
\affaddr{\affmark[2]\textit{Concordium Blockchain Research Center, Aarhus University, Denmark}}\\ \\
E-mail: \{deuber,thyagarajan\}@cs.fau.de\\ magri@cs.au.dk
}
\date{\today}
\fi
\begin{document}

\maketitle

\iffull
\thispagestyle{empty}
\input{abstract.tex}
\newpage
\clearpage
\else
\input{abstract.tex}
\fi

\iffull
\tableofcontents
\thispagestyle{empty}
\newpage
\setcounter{page}{1}
\fi

\ifieeesp
\renewcommand{\thefootnote}{\fnsymbol{footnote}}
\footnotetext[2]{Work done while the author was affiliated with Friedrich-Alexander University Erlangen-N\"urnberg.}
\renewcommand{\thefootnote}{\arabic{footnote}}
\fi

\ifieeesp
\begin{IEEEkeywords}
    Blockchain, Bitcoin, Redactable Blockchain, GDPR
\end{IEEEkeywords}
\fi

\input{introduction.tex}

\input{preliminaries.tex}

\input{protocol.tex}

\input{analysis.tex}

\input{bitcoin.tex}

\input{impl.tex}

\input{attacks.tex}

\section*{Acknowledgment}

  This work is a result of the collaborative research project PROMISE
  (16KIS0763) by the German Federal
  Ministry of Education and Research (BMBF).
  FAU authors were also supported by the German research foundation
  (DFG) through the collaborative research center 1223, and by the
  state of Bavaria at the Nuremberg Campus of Technology (NCT). NCT is
  a research cooperation between the Friedrich-Alexander-Universit\"at
  Erlangen-N\"urnberg (FAU) and the Technische Hochschule N\"urnberg
  Georg Simon Ohm (THN).

\iffull
\bibliographystyle{plain}
	\bibliography{cryptobib/abbrev1,cryptobib/crypto,cryptobib/extrarefs}
	\input{appendix.tex}
\fi

\ifieeesp
\bibliographystyle{IEEEtranS}
	\bibliography{cryptobib/abbrev3,cryptobib/crypto,cryptobib/extrarefs}
	\input{appendix.tex}
\fi
 
%\bibliographystyle{ieeetr}

%\bibliography{abbrev2,crypto_crossref,extrarefs}

\end{document}

%% file: macro.tex
\newtheorem{theorem}{Theorem}
\newtheorem{definition}{Definition}

\mathchardef\mhyphen="2D

\newcommand{\NN}{\mathbb{N}}

\newcommand{\spar}{\kappa}

\newcommand{\bit}{\{0,1\}}

\newcommand{\advA}{\mathcal{A}}

\newcommand{\ctr}{\mathit{ctr}}
\newcommand{\round}{r}
\newcommand{\chain}{\mathcal{C}}
\newcommand{\head}[1]{\mathsf{Head}(#1)}

\newcommand{\block}{\langle s,x,\ctr \rangle}
\newcommand{\blockmod}{\langle s, x,\ctr,y \rangle}
\newcommand{\blockmodi}[1]{\langle s_{#1},x_{#1},\ctr_{#1},y_{#1} \rangle}
\newcommand{\blockprime}{\langle s',x',\ctr' \rangle}

\newcommand{\length}[1]{\mathsf{len}(#1)}

\newcommand{\prune}[2]{#1^{\lceil #2}}
\newcommand{\pruneback}[2]{{}^{#2\rceil}#1}

\newcommand{\policy}{\mathcal{P}}
\newcommand{\redactTx}{\mathsf{editTx}}
\newcommand{\slot}{\mathit{sl}}

\newcommand{\redactreq}{\mathsf{proposeEdit}}

\newcommand{\votefunc}{\mathsf{Vt}}
\newcommand{\decide}{\mathsf{validateCand}}
\newcommand{\decideext}{\mathsf{validateCandExt}}

\newcommand{\rcbpool}{\mathcal{R}}

\newcommand{\genesis}{\mathsf{genesis}}
\newcommand{\voting}{\mathtt{voting}}
\newcommand{\accept}{\mathtt{accept}}
\newcommand{\reject}{\mathtt{reject}}
\newcommand{\candidateblk}{B^\star}
\newcommand{\editreq}{\ensuremath{B^\star}}
\newcommand{\id}{\mathsf{ID}}

\newcommand{\Blk}{\ensuremath{\Gamma}}

\newcommand{\chainValid}{\ensuremath{\mathsf{validateChain}}}
\newcommand{\chainValidext}{\ensuremath{\mathsf{validateChainExt}}}
\newcommand{\blockValid}{\ensuremath{\mathsf{validateBlock}}}
\newcommand{\blockValidext}{\ensuremath{\mathsf{validateBlockExt}}}

\newcommand{\getChain}{\ensuremath{\mathsf{updateChain}}}

\newcommand{\broadcast}{\ensuremath{\mathsf{broadcast}}}
\newcommand{\tx}{\mathsf{Tx}}

\newcommand{\data}{\mathit{x}}

\newcommand{\amount}{\alpha}

\newcommand{\env}{\mathcal{Z}}
\newcommand{\adv}{\mathcal{A}}
\newcommand{\view}{\mathsf{view}}

\newcommand{\userSet}{\mathcal{U}}

%% file: abstract.tex
\begin{abstract}
	Bitcoin is an immutable permissionless blockchain system that has been extensively used as a public bulletin board by many different applications that heavily relies on its immutability. However, Bitcoin's immutability is not without its fair share of demerits. Interpol exposed the existence of harmful and potentially illegal documents, images and links in the Bitcoin blockchain, and since then there have been several qualitative and quantitative analysis on the types of data currently residing in the Bitcoin blockchain. 
	Although there is a lot of attention on blockchains, surprisingly the previous solutions proposed for data redaction in the permissionless setting are far from feasible, and require additional trust assumptions. Hence, the problem of harmful data still poses a huge challenge for law enforcement agencies like Interpol (Tziakouris, IEEE S\&P'18).
	
	We propose the first efficient redactable blockchain for the permissionless setting that is easily integrable into Bitcoin, and that does not rely on heavy cryptographic tools or trust assumptions. Our protocol uses a consensus-based voting and is parameterised by a policy that dictates the requirements and constraints for the redactions; if a redaction gathers enough votes the operation is performed on the chain. As an extra feature, our protocol offers public verifiability and accountability for the redacted chain. Moreover, we provide formal security definitions and proofs showing that our protocol is secure against redactions that were not agreed by consensus. Additionally, we show the viability of our approach with a proof-of-concept implementation that shows only a tiny overhead in the chain validation of our protocol when compared to an immutable one.

\end{abstract}

%% file: introduction.tex
\section{Introduction}

%\NEWA{ TODO list by priority
%\begin{enumerate}
%
%	\item \st{Security - Theorems and proofs (abstraction) (Bernardo and Aravind)}
%	\item Instantiation - Bitcoin, Accountability, Future consistency (Draft done - Aravind)
%	\item Implementation - Evaluation/Benchmarking (DD)
%	\item Intro - motivation, related work, discuss trivial solutions, high level  overview (Bernardo)
%	\item Discussion - attacks (Aravind and DD)
%\end{enumerate}
%
%}
%==DRAFT==
%  \TODOB{Capitalize the first letters in the subsections title}

 Satoshi Nakamoto's 2008 proposal of Bitcoin~\cite{nakamoto2008bitcoin} has revolutionised the financial sector. It helped realise a monetary system without relying on a central trusted authority, which has since then given rise to hundreds of new systems known as cryptocurrencies. Interestingly however, a closer look into the basics of Bitcoin sheds light on a new technology, blockchains. Ever since, there has been a lot of ongoing academic research~\cite{EC:GarKiaLeo15,kiayias2017ouroboros,cryptoeprint:2018:378,breidenbach2018enter} on the security and applications of blockchains as a primitive. A blockchain in its most primitive form is a decentralised chain of agreed upon blocks containing timestamped data.  %For instance, in Bitcoin transactions between users are the timestamped data.
 
  A consensus mechanism supports the decentralised nature of blockchains. There are different types of consensus mechanisms that are based on different resources, such as Proof of Work (PoW) based on computational power, Proof of Stake (PoS) based on the stake in the system, Proof of Space based on storage capacity, among many others. Typically, users in the system store a local copy of the blockchain and run the consensus mechanism to agree on a unified view of the blockchain. These mechanisms must rely on non-replicability of resources to be resilient against simple sybil attacks where the adversary spawn multiple nodes under his control. 
  %Bitcoin exploits this immutability property by not allowing tampering of the monetary information to emulate a public financial ledger. 
  
  Apart from its fundamental purpose of being a digital currency, Bitcoin exploits the properties of its blockchain, as in being used as a tool for many different applications, such as timestamp service~\cite{gipp2015decentralized,gipp2016securing}, to achieve fairness and correctness in secure multi-party computation~\cite{SP:ADMM14,FCW:ADMM14,C:BenKum14,CCS:KumBen14}, and to build smart contracts~\cite{SP:KMSWP16}. It acts as an immutable public bulletin board, supporting the storage of arbitrary data through special operations. For instance, the $\mathtt{OP\_RETURN}$ code, can take up to 80 bytes of arbitrary data that gets stored in the blockchain. With no requirement for centralised trust and its capability of supporting complex smart contracts, communication through the blockchain has become practical, reasonably inexpensive and very attractive for applications. 
   
 %  However, non-tampering of the data entry does not depend on the entry itself as shown in the work of Matzutt et al.~\cite{matzutt2018quantitative}. They provide a quantitative analysis showing that the Bitcoin blockchain is a host to several objectionable and sensitive contents, such as child pornography, nude picture of a young woman, wikileaks documents, stolen private keys, etc. 
  
% \TODOB{Explain before what ``stabilised in the chain" means.} 
 
\paragraph{Blockchain and Immutability} 
The debate about the immutability of blockchain protocols has gained worldwide attention lately due to the adoption of the new General Data Protection Regulation (GDPR) by European states. Several provisions of the GDPR regulation are inherently incompatible with current permissionless immutable blockchain proposals (e.g.,\ Bitcoin and Ethereum)~\cite{ibanez2018blockchains} as it is not possible to remove any data (addresses, transaction values, timestamp information) that has stabilised\footnote{A transaction (or data) is considered stable in the blockchain when it is ``deep" enough into the chain. We formally define this property in~\cref{sec:blockchain_properties}.} in the chain in such protocols. Since permissionless blockchains are completely decentralised and allow for any user to post transactions to the chain for a small fee, malicious users can post transactions to the system containing illegal and/or harmful data, such as (child) pornography, private information or stolen private keys, etc. The existence of such illicit content was first reported in~\cite{interpol2015} and has remained a challenge for law enforcement agencies like Interpol~\cite{tziakouris2018cryptocurrencies}. Moreover, quantitative analysis in the recent work of Matzutt et al.~\cite{matzutt2018quantitative} shows that it is not feasible to ``filter" all data from incoming transactions to check for malicious contents before the transaction is inserted into the chain. Therefore, once it becomes public knowledge that malicious data was inserted (and has stabilised) into the chain, the honest users are faced with the choice of either, willingly broadcast illicit (and possibly illegal~\cite{matzutt2018quantitative,coindesk2018}) data to other users, or to stop using the system altogether.

This effect greatly hinders the adoption of permissionless blockchain systems, as honest users that are required to comply with regulations, such as GDPR, are forced to withdraw themselves from the system if there is no recourse in place to deal with illicit data inserted into the chain.

%  
%  \NEWA{
%  Points for intro-
%
%1) \st{Purpose of blockchain - timestamping}
%
%2) \st{Bitcoin - immutable ledger}
%
%3) Misuse - illicit content, law enforcement, legal compliance, GDPR
%
%4) \st{State of the art - permissioned chains, trust assumption on miners, compliance from transaction generator}
%
%5) no feasible decentralised solution for permission less setting
%  
%  }
  
  \subsection{State of the Art}
Specifically to tackle the problem of arbitrary harmful data insertions in the blockchain, the notion of redacting the contents of a blockchain was first proposed by Ateniese et al.~\cite{ateniese2017redactable}. The authors propose a solution more focused on the permissioned blockchain setting\footnote{The permissioned blockchain setting is when there is a trusted third party (TTP) that deliberates on the users' entry into the system.} based on chameleon hashes~\cite{camenisch2017chameleon}. In their protocol, a chameleon hash function replaces the regular SHA256 hash function when linking consecutive blocks in the chain. When a block is modified, a collision for the chameleon hash function can be efficiently computed (with the knowledge of the chameleon trapdoor key) for that block, keeping the state of the chain consistent after arbitrary modifications. 

In a permissioned setting where the control of the chain is shared among a few semi-trusted parties, the solution from~\cite{ateniese2017redactable} is elegant and works nicely, being even commercially adopted by a large consultancy company~\cite{nytimes, businessinsider, ateniese2018rewritable}. However, in permissionless blockchains such as Bitcoin, where the influx of users joining and leaving the system is ever changing and without any regulation, their protocol  clearly falls short in this scenario, as their techniques of secret sharing the chameleon trapdoor key and running a MPC protocol to compute a collision for the chameleon hash function do not scale to the thousands of users in the Bitcoin network. 
  Moreover, when a block is removed in their protocol it is completely unnoticeable to the users, leaving no trace of the old state. Although this could make sense in a permissioned setting, in a permissionless setting one would like to have some public accountability as to when and where a redaction has occurred. 
  
%Redactable blockchains was proposed in 2017 by Ateniese et al.~\cite{ateniese2017redactable}, where miners could come together to delete blocks with harmful contents. Their protocol works in a permissioned setting where there is a Trusted Third Party (TTP) that decides on participants' entry into the system. 

%They use Chameleon hashes~\cite{camenisch2017chameleon} to replace the regular SHA256 to link blocks. Essentially, each miner in the system holds a share of the trapdoor key of the chameleon hash. When the miners decide to redact a block $B_j$, reconstruct the trapdoor key and find a collision for the hash of $B_{j-1}$. They then remove the block $B_j$ and let block $B_{j+1}$ point to $B_{j-1}$. However, the proposal is tailored for a permissioned setting and for a permission-less setting such as a Bitcoin, it is unclear how to secret share the trapdoor key efficiently among a large ever-changing set of participants. Moreover, the scalability of running a large scale MPC in the Bitcoin network to reconstruct the trapdoor key is a bottleneck for the protocol's efficiency. In their proposal, the entire block is deleted leaving no trace of the old state. Meaning, a new user who gets a new chain will have no clue if the chain was redacted or not. Although this could make sense in a permissioned setting, in a permission-less setting one would like to have some public accountability of when, where and what can and was changed (or removed). 

  Later, Puddu et al.~\cite{puddu2017muchain} proposed a blockchain protocol where the sender of a transaction can encrypt alternate versions of the transaction data, known as ``mutations"; the only unencrypted version of the transaction is considered to be the active transaction. The decryption keys are secret shared among the miners, and the sender of a transaction establishes a mutation policy for his transaction, that details how (and by whom) his transaction is allowed to be mutated. On receiving a mutate request, the miners run a MPC protocol to reconstruct the decryption key and decrypt the appropriate version of the transaction. The miners then publish this new version as the active transaction.
   In case of permissionless blockchains, they propose the usage of voting for gauging approval based on computational power. However, in a permissionless setting a malicious user can simply not include a mutation for his transaction, or even set a mutation policy where only he himself is able to mutate the transaction. Moreover, to tackle transaction consistency, where a mutated transaction affects other transactions in the chain, they propose to mutate all affected transactions through a cascading effect. This however, completely breaks the notion of transaction stability, e.g.,\ a payment made in the past to a user could be altered as a result of this cascading mutation. The proposal of~\cite{puddu2017muchain} also suffers from scalability issues due to the MPC protocol used for reconstructing decryption keys across different users.
   
   It is clear that for a permissionless blockchain without centralised trust assumptions, a practical solution for redacting harmful content must refrain from employing large-scale MPC protocols that hinders the performance of the blockchain. It also must accommodate public verifiability and accountability such that rational miners are incentivised to follow the protocol.

\subsection{Our Contributions}

\paragraph{Editable Blockchain Protocol}
  	We propose the first editable blockchain protocol for permissionless systems in~\cref{sec:editing}, which is completely decentralised and does not rely on heavy cryptographic primitives or additional trust assumptions. This makes our protocol easily integrable in systems like Bitcoin (as described in~\cref{sec:instantiation}). The edit operations can be proposed by any user and they are voted in the blockchain through consensus; the edits are only performed if approved by the blockchain policy (e.g.,\ voted by the majority). The protocol is based on a PoW consensus, however, it can be easily adapted to any consensus mechanism, since the core ideas are inherently independent of the type of consensus used. Our protocol also offers accountability for edit operations, where any edit in the chain can be publicly verified. 
  	
\paragraph{Formal Analysis} We build our protocol on firm theoretical grounds, as we formalise all the necessary properties of an editable blockchain in~\cref{sec:security}, and later show that our generic protocol of~\cref{sec:protocol} satisfies these properties. We borrow the fundamental properties of a secure blockchain protocol from~\cite{EC:GarKiaLeo15} and adapt them to our setting. 

%\paragraph{Formal Analysis} We formally show that our protocol satisfies persistence and liveness upto accountable editing. For this we consider the elementary security properties of a blockchain protocol from~\cite{EC:GarKiaLeo15,kiayias2017ouroboros}, namely, \emph{chain quality, chain growth} and \emph{common prefix}. We extend these properties to account for editing operations and the editing policy of our editable blockchain protocol. In this regard we propose a new definition \emph{editable common prefix} and formally prove that our protocol satisfies these extended security definitions.

\paragraph{Implementation} We demonstrate the practicality of our protocol with a proof-of-concept implementation in Python. We first show in~\cref{sec:implementation} that adding our redaction mechanism incurs in just a small overhead for chain validation time compared to that of the immutable protocol. Then, we show that for our protocol the overhead incurred for different numbers of redactions in the chain against a redactable chain with no redactions is minimal (less than $3\%$ for $5,000$ redactions on a $50,000$ blocks chain). Finally, we analyse the effect of the parameters in our protocol by measuring the overhead introduced by different choices of the system parameters when validating chains with redactions.

\subsection{Our Protocol}
  
Our protocol extends the immutable blockchain of Garay et al.~\cite{EC:GarKiaLeo15} to accommodate for edit operations in the following way: We extend the block structure to accommodate another copy of the transaction's Merkle root, that we denote by \emph{old state}. We also consider an editing policy for the chain, that determines the constraints and requirements for approving edit operations. To edit a block in the chain, our protocol (\cref{fig:abstract_protocol}) executes the following steps:
  	\begin{enumerate}
  		\item[a)] A user first proposes an edit request to the system. The request consists of the index of the block he wants to edit, and a candidate block to replace it.
  		
  		\item[b)] When miners in the network receives an edit request, they first validate the candidate block using its \emph{old state} information and verifying the following conditions: (1) it contains the correct information about the previous block, (2) it has solved the proof of work and (3) it does not invalidate the next block in the chain. If the candidate block is valid, miners can vote for it during the request's voting period by simply including the hash of the request in the next block they mine.  The collision resistance property of the hash function ensures that a vote for an edit request cannot be considered as a vote for any other edit request. 
  		
  		\item[c)] After the voting period for a request is over, everyone in the network can verify if the edit request was approved in accordance to the policy (e.g.,\ by checking the number of votes it received). If the request was approved, then the edit operation is performed by replacing the original block with the candidate block.
  	\end{enumerate}
  	
 To validate an edited chain, the miners validate each block exactly like in the immutable protocol; if a ``broken" link is found between blocks, the miner checks if the link still holds for the \emph{old state} information\footnote{A similar technique is used in~\cite{ateniese2017multiple}  to ``scar'' a block that was previously redacted.}. In the affirmative case, the miner ensures that the edited block has gathered enough votes and is approved, according to the policy of the chain. 
 
 The process of a redaction in our generic protocol as described in~\cref{fig:protocol} is pictorially presented in~\cref{fig:abstract_protocol}.

\input{picture.tex}

% a block has been edited by simply checking if the hash links do not hold for the existing state of the block but rather hold for the \emph{old state} information. All that is left for the miner to do is to ensure that the edited block has gathered enough votes and is approved according the editing policy of the chain. 

%\paragraph{Implementation} 
  
%  \TODOA{Talk about implementation}

%% file: picture.tex
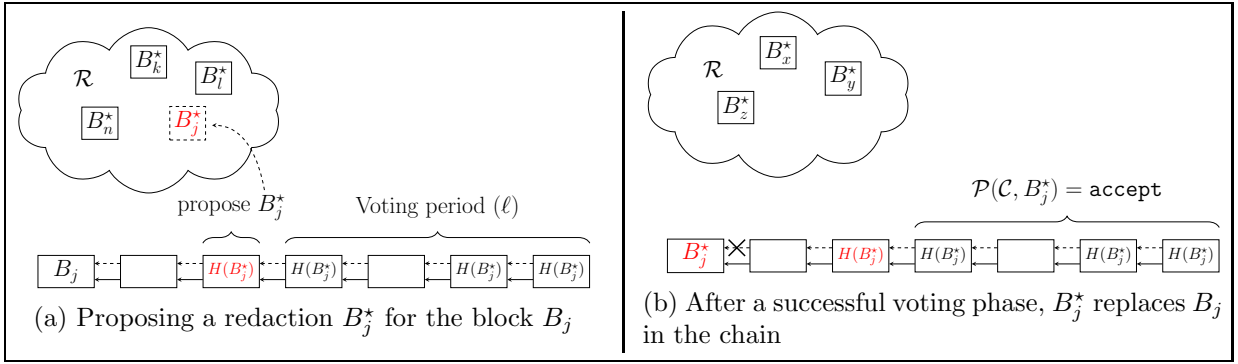
\begin{figure*}[t]
\begin{minipage}{0.48\linewidth}
	\centering
	\resizebox{\linewidth}{!}{
	\begin{tikzpicture}[every node/.style={font=\tiny}]
		
	    \node (n1) []  {};
	    
		%boxes at the top	
		\node (b1) [block, below of = n1,minimum width=4.5em] {\huge$B_j$};
		\node (b2) [block, right of = b1, xshift=4em, minimum width=4.5em] {};
		\node (b3) [block, right of = b2, xshift=4em,minimum width=4.5em] {\Large \textcolor{red}{$H(\candidateblk_j)$}};
		\node (b4) [block, right of = b3, xshift=4em,minimum width=4.5em] {\Large $H(\candidateblk_j)$};
		\node (b5) [block, right of = b4, xshift=4em,minimum width=4.5em] {};
		\node (b6) [block, right of = b5, xshift=4em,minimum width=4.5em] {\Large $H(\candidateblk_j)$};
		\node (b7) [block, right of = b6, xshift=4em,minimum width=4.5em] {\Large $H(\candidateblk_j)$};

		\draw [arrow,dashed] ([yshift=0.25cm]b2.west) -- ([yshift=0.25cm]b1.east);
		
		\draw [arrow] ([yshift=-0.25cm]b2.west) -- ([yshift=-0.25cm]b1.east);
		\draw [arrow,dashed] ([yshift=0.25cm]b3.west) -- ([yshift=0.25cm]b2.east);
		
		\draw [arrow] ([yshift=-0.25cm]b3.west) -- ([yshift=-0.25cm]b2.east);
		
		\draw [arrow,dashed] ([yshift=0.25cm]b4.west) -- ([yshift=0.25cm]b3.east);
		\draw [arrow] ([yshift=-0.25cm]b4.west) -- ([yshift=-0.25cm]b3.east);
			\draw [arrow,dashed] ([yshift=0.25cm]b5.west) -- ([yshift=0.25cm]b4.east);
		
		\draw [arrow] ([yshift=-0.25cm]b5.west) -- ([yshift=-0.25cm]b4.east);
		\draw [arrow,dashed] ([yshift=0.25cm]b6.west) -- ([yshift=0.25cm]b5.east);
		\draw [arrow] ([yshift=-0.25cm]b6.west) -- ([yshift=-0.25cm]b5.east);
		\draw [arrow,dashed] ([yshift=0.25cm]b7.west) -- ([yshift=0.25cm]b6.east);
		\draw [arrow] ([yshift=-0.25cm]b7.west) -- ([yshift=-0.25cm]b6.east);

		%text at the top
		\node (nt1) at ($(b4.west)!0.5!(b7.east)$) {};

\draw [decorate,decoration={brace,amplitude=10pt,raise=25pt},xshift=0pt,yshift= 4cm]
(b4.west) -- (b7.east) node [black,midway,yshift=2cm] 
{\huge Voting period ($\ell$)};

\draw [decorate,decoration={brace,amplitude=10pt,raise=25pt},yshift= 4cm]
(b3.west) -- (b3.east) node (propose) [black,midway,yshift=2cm] 
{\huge propose $\candidateblk_j$};

%\draw [decorate,decoration={brace,amplitude=10pt,raise=25pt},yshift= 4cm]
%(b2.west) -- (b2.east) node [black,midway,yshift=2cm] 
%{\huge Block $B_j$};

\node [cloud, above of =b2, yshift = 4cm,draw,cloud puffs=10,cloud puff arc=120, aspect=2, inner ysep=3.6em](cl) { };
\node (intro4) at ($(cl.east)$){};
%\draw[arrow, line width=1.2pt] (intro2) to[out=-160,in=0] (intro4);

\node (cl1) [block, above of =cl, yshift =0.5cm]{\huge $B^\star_k$};
\node (cl2) [block, right of =cl1,xshift = 1cm,yshift =-0.5cm]{\huge $B^\star_l$};
\node (cl3) [left of =cl1,xshift = -1cm,yshift =-0.5cm]{\huge $\rcbpool$};
\node (cl4) [block, left of =cl,xshift = -0.5cm,yshift =-0.4cm]{\huge $B^\star_n$};
\node (cl5) [block, dashed, right of =cl,xshift = 0.2cm,yshift =-0.4cm]{\huge \textcolor{red}{$\candidateblk_j$}};

\draw [arrow,dashed] ([xshift=-2em]propose.north east) to[out=100,in=0] ([yshift=1.2em,xshift=0.5em]cl5.south east);

	\end{tikzpicture}
	}
\subcaption{Proposing a redaction $\candidateblk_j$ for the block $B_j$}
\label{fig:abstract_before_redaction}
\end{minipage}
\hfill\vline\hfill
\begin{minipage}{0.48\linewidth}
	\centering
	\resizebox{\linewidth}{!}{
	\begin{tikzpicture}[every node/.style={font=\tiny}]
		
	    \node (n1) []  {};
	    
		%boxes at the top	
		\node (b1) [block, below of = n1,minimum width=4.5em] {\huge \textcolor{red}{$B^\star_j$}};
		\node (b2) [block, right of = b1, xshift=4em,minimum width=4.5em] {};
		\node (b3) [block, right of = b2, xshift=4em, minimum width=4.5em] {\Large \textcolor{red}{$H(\candidateblk_j)$}};
		\node (b4) [block, right of = b3, xshift=4em,minimum width=4.5em] {\Large $H(\candidateblk_j)$};
		\node (b5) [block, right of = b4, xshift=4em,minimum width=4.5em] {};
		\node (b6) [block, right of = b5, xshift=4em,minimum width=4.5em] {\Large $H(\candidateblk_j)$};
		\node (b7) [block, right of = b6, xshift=4em,minimum width=4.5em] {\Large $H(\candidateblk_j)$};
		
		\node (cross) [cross out, inner sep = 0.6em,draw=black,yshift=0.25cm,line width = 1.3pt] at ($(b1.east)!0.5!(b2.west)$){};
		%\node (b8) [block, right of = b7, xshift=2cm] {};
		%\node (b9) [block, right of = b8, xshift=2cm] {};
		%dotted lines / arrows between the boxes
		%\draw [arrow, dashed, very thick] (b2) -- (b1);
		\draw [arrow,dashed] ([yshift=0.25cm]b2.west) -- ([yshift=0.25cm]b1.east);
		
		\draw [arrow] ([yshift=-0.25cm]b2.west) -- ([yshift=-0.25cm]b1.east);
		\draw [arrow,dashed] ([yshift=0.25cm]b3.west) -- ([yshift=0.25cm]b2.east);
		
		\draw [arrow] ([yshift=-0.25cm]b3.west) -- ([yshift=-0.25cm]b2.east);
		%%\draw [arrow,dashed] ([yshift=0.25cm]b4.west) -- ([yshift=0.25cm]b3.east);
		
		\draw [arrow,dashed] ([yshift=0.25cm]b4.west) -- ([yshift=0.25cm]b3.east);
		\draw [arrow] ([yshift=-0.25cm]b4.west) -- ([yshift=-0.25cm]b3.east);
			\draw [arrow,dashed] ([yshift=0.25cm]b5.west) -- ([yshift=0.25cm]b4.east);
		
		\draw [arrow] ([yshift=-0.25cm]b5.west) -- ([yshift=-0.25cm]b4.east);
		\draw [arrow,dashed] ([yshift=0.25cm]b6.west) -- ([yshift=0.25cm]b5.east);
		\draw [arrow] ([yshift=-0.25cm]b6.west) -- ([yshift=-0.25cm]b5.east);
		\draw [arrow,dashed] ([yshift=0.25cm]b7.west) -- ([yshift=0.25cm]b6.east);
		\draw [arrow] ([yshift=-0.25cm]b7.west) -- ([yshift=-0.25cm]b6.east);

		%text at the top
		\node (nt1) at ($(b4.west)!0.5!(b7.east)$) {};

\draw [decorate,decoration={brace,amplitude=10pt,raise=25pt},xshift=0pt,yshift= 4cm]
(b4.west) -- (b7.east) node [black,midway,yshift=2cm] 
{\huge $\policy(\chain, B^\star_j) = \accept$};

\node [cloud, above of =b2, yshift = 4cm,draw,cloud puffs=10,cloud puff arc=120, aspect=2, inner ysep=3.6em](cl) { };
\node (cl1) [block, above of =cl, yshift =0.3cm]{\huge $B^\star_x$};
\node (cl2) [block, right of =cl1,xshift = 1cm,yshift =-0.8cm]{\huge $B^\star_y$};
\node (cl3) [left of =cl1,xshift = -1cm,yshift =-0.5cm]{\huge $\rcbpool$};
\node (cl4) [block, left of =cl,xshift = -0.3cm,yshift =-0.4cm]{\huge $B^\star_z$};

	\end{tikzpicture}
	}
\subcaption{After a successful voting phase, $B^\star_j$ replaces $B_j$ in the chain}
\label{fig:abstract_after_redaction}
\end{minipage}

\caption{The candidate block pool $\rcbpool$ stores the candidate blocks that are proposed and that can be endorsed in the voting phase. A block is linked to its predecessor by two links, the old link (solid arrow) and the new link (dashed arrow). In (a), a redact request $B^\star_j$ is proposed as a redaction for $B_j$ and added to $\rcbpool$, then the hash of $B^\star_j$ is included in the chain to denote a new candidate redaction; its voting phase starts just after its proposal. In (b), the candidate block $B^\star_j$ has gathered enough votes and was approved by the redaction policy $\policy$ of the chain; $B^\star_j$ replaces $B_j$ and the redacted chain is propagated. Note that new link from the block to the right of $B^\star_j$ is broken (marked by a cross), however the old link to $B^\star_j$ still holds. For simplicity, we consider the parameters $k=0$ (persistence), $\ell=4$ (voting period) and $\rho\ge 3/4$ (threshold for policy approval).}
\label{fig:abstract_protocol}
\end{figure*}

%% file: preliminaries.tex
\section{Preliminaries}

%\subsection{Notations}
Throughout this work we denote by $\spar \in \NN$ the security parameter and by $a \gets \adv (\mathsf{in})$ the output of an algorithm $\adv$ on input $\mathsf{in}$. We also use the terms ``redact" and ``edit" interchangeably in this paper.
\subsection{Blockchain Basics}\label{sec:block-abstract}
%\Dnote{Concise description of the blockchain using the notation from Aggelos' paper.}

We make use of the notation of~\cite{EC:GarKiaLeo15} to describe a blockchain. A
block is a triple of the form $B := \block$, where $s \in \bit^\spar$,
$x \in \bit^*$ and $\ctr \in \NN$. Here $s$ is the state of the previous block, $x$ is the data and $\ctr$ is the proof of work of the block. A block $B$ is \emph{valid} iff
\[
\blockValid^D(B):=H(\ctr,G(s,x)) < D.
\]
%\TODOB{We don't need this parameter $q$ do we?}

Here, $H:\bit^*\rightarrow\bit^\spar$ and $G:\bit^*\rightarrow\bit^\spar$ are
cryptographic hash functions, and the parameter $D \in \NN$ is the block's difficulty level.

The blockchain is simply a chain (or sequence) of blocks, that we call
$\chain$. The rightmost block is called the head of the chain, denoted by
$\head{\chain}$. Any chain $\chain$ with a head $\head{\chain} := \block$ can be
extended to a new longer chain $\chain':=\chain||B'$ by attaching a (valid)
block $B' := \blockprime$ such that $s' = H(\ctr,G(s,x))$; the head of the new
chain $\chain'$ is $\head{\chain'} := B'$. A chain $\chain$ can also be empty,
and in such a case we let $\chain := \varepsilon$. The function
$\length{\chain}$ denotes the length  of a chain $\chain$ (i.e.,\ its number of blocks). For
a chain $\chain$ of length $n$ and any $q \ge 0$, we denote by
$\prune{\chain}{q}$ the chain resulting from removing the $q$ rightmost blocks
of $\chain$, and analogously we denote by $\pruneback{\chain}{q}$ the chain
resulting in removing the $q$ \emph{leftmost} blocks of $\chain$; note that if $q \ge n$ (where $\length{\chain} = n$)
then $\prune{\chain}{q} :=\varepsilon$ and $\pruneback{\chain}{q}:=\varepsilon$. If
$\chain$ is a prefix of $\chain'$ we write $\chain \prec \chain'$. We also note
that the difficulty level $D$ can be different among blocks in a chain.
%
%%The functions $V(\cdot), I(\cdot), R(\cdot)$ presented in~\cite{GarayKL15} are
%%application specific, and therefore are independent of the blockchain structure.
%%The function $V(\cdot)$ determines the structure of the information that is
%%stored into the blockchain, the function $I(\cdot)$ specifies how the contents
%%of a block are created by the participants, and function $R(\cdot)$ tells how
%%the application should interpret the blockchain.
%
%The work of~\cite{GarayKL15} models the Bitcoin protocol in a setting where the
%number of participants is always fixed and the network in synchronized. They show
%that the protocol satisfies \emph{consistency} in this model, meaning that all
%honest participants have the same chain prefix of the blockchain. A more recent
%work by Pass, Seeman and shelat~\cite{PassSS16} analyzes the case where the
%network is asynchronous and the number of participants can dynamically change.
%%,and they also show a positive result for the chain consistency. 
%We point out that
%our framework is independent of the network type in these models.

\subsection{Properties of a Secure Blockchain}\label{sec:blockchain_properties}
In this section we detail the relevant aspects of the underlying blockchain system that is required for our protocol.

 We consider time to be divided into standard discrete units, such as minutes. A well defined continuous amount of these units is called a \emph{slot}. Each slot $\slot_l$ is indexed for $l \in \{1,2,3, \ldots \}$. We assume that users have a synchronised clock that indicates the current time down to the smallest discrete unit. The users execute a distributed protocol to generate a new block in each slot, where a block contains some data. We assume the slots' real time window properties as in~\cite{kiayias2017ouroboros}. In~\cite{EC:GarKiaLeo15,EC:PassSS17,kiayias2017ouroboros} it is shown that a ``healthy'' blockchain must satisfy the properties of \emph{persistence} and \emph{liveness}, which intuitively guarantee that after some time period, all honest users of the system will have a consistent view of the chain, and transactions posted by honest users will eventually be included. We informally discuss the two properties next.

	\medskip

	\noindent\emph{Persistence}: Once a user in the system announces a particular transaction as \emph{stable}, all of the remaining users when queried will either report the transaction in the same position in the ledger or will not report any other conflicting transaction as stable. A system parameter $k$ determines the number of blocks that stabilise a transaction. That is, a transaction is stable if the block containing it has at least $k$ blocks following it in the blockchain. We only consider a transaction to be in the chain after it becomes stable.

	\medskip
		
	\noindent\emph{Liveness}: If all the honest users in the system attempt to include a certain transaction into their ledger, then after the passing of time corresponding to $u$ slots which represents the \emph{transaction confirmation time}, all users, when queried and responding honestly, will report the transaction as being stable.

	\medskip
	
Throughout the paper we refer to the user as both a user and a miner interchangeably.
%	As shown in~\cite{EC:GarKiaLeo15,EC:PassSS17,kiayias2017ouroboros} the following three properties help achieve persistence and liveness.
%	\begin{definition}[Common Prefix, $k \in \NN$]\label{def:common_prefix}
%		The chains $\chain_1,\chain_2$ possessed by two honest parties at the onset of the slots $\slot_1 < \slot_2$ are such that $\prune{\chain_1}{k} \preceq \chain_2$, where $\prune{\chain_1}{k}$ denotes the chain obtained by removing the last $k$ blocks from $\chain_1$.
%	\end{definition}
%	\begin{definition}[Chain Quality, $0 <\mu \le 1, \ell \in \NN$]
%		Consider any portion of length at least $\ell$ of the chain possessed by an honest party at least $\ell$ of the chain possessed by an honest party at the onset of a round; the ratio of blocks originating from the adversary is at most $\mu$. Here $\mu$ is the chain quality coefficient.		
%	\end{definition}
%	\begin{definition}[Chain Growth, $0<\tau \le 1, s \in \NN$]
%		Consider the chains $\chain_1,\chain_2$ possessed by two honest parties at the onset of two slots $\slot_1,\slot_2$ with $\slot_2$ at least $s$ slots ahead of $\slot_1$. Then it holds that $\length{\chain_2}-\length{\chain_1} \ge \tau \cdot s$. Here $\tau$ is the speed coefficient.	
%	\end{definition}

\subsection{Execution Model.}
In the following we define the notation for our protocol executions. Our definitions follow along the same lines of~\cite{pass2017fruitchains}.

A protocol refers to an algorithm for a set of interactive Turing Machines (also called nodes) to interact with each other. The execution of a protocol $\Pi$ that is directed by an environment/outer game $\env(1^\spar)$, which activates a number of parties $\userSet = \{p_1, \ldots, p_n \}$ as either honest or corrupted parties. Honest parties would faithfully follow the protocol's prescription, whereas corrupt parties are controlled by an adversary $\adv$, which reads all their inputs/messages and sets their outputs/messages to be sent.

\begin{itemize}
	\item A protocol's execution proceeds in rounds that model atomic time steps. At the beginning of every round, honest parties receive inputs from an environment $\env$; at the end of every round, honest parties send outputs to the environment $\env$.

	\item $\adv$ is responsible for delivering all messages sent by parties (honest or corrupted) to all other parties. $\adv$ cannot modify the content of messages broadcast by honest parties.
	
	\item At any point $\env$ can corrupt an honest party $j$, which means that $\adv$ gets access to its local state and subsequently controls party $j$. 
	%(In particular, this means we consider a model with \emph{erasures}: Random coin tosses that are no longer stored in the local state of $j$ and therefore are not visible to $\adv$.)
	
	\item At any point of the execution, $\env$ can uncorrupt a corrupted party $j$, which means that $\adv$ no longer controls $j$. A party that becomes uncorrupt is treated in the same way as a newly spawning party, i.e., the party's internal state is re-initialised and then the party starts executing the honest protocol no longer controlled by $\adv$.
\end{itemize}
Note that a protocol execution can be randomised, where the randomness comes from honest parties as well as from $\adv$ and $\env$. We denote by $\view \gets \mathsf{EXEC}^\Pi(\adv, \env, \spar)$ the randomly sampled execution trace. More formally, $\view$ denotes the joint view of all parties (i.e., all their inputs, random coins and messages received, including those from the random oracle) in the above execution; note that this joint view fully determines the execution. 
%For convenience, we denote by $\view_\delta$ the state of the execution at round $\delta$ and we define the following two functions: On input an execution $\view_\delta$ at a certain round $\delta$, the function $\public$ returns the public state of the protocol, i.e., the information available to all machines, whereas $\honest$ returns the set of honest parties at round $\delta$.

%% file: protocol.tex
\section{Editing the Blockchain}\label{sec:editing}
In this section we introduce an abstraction $\Blk$ of a blockchain protocol, and we describe how to extend $\Blk$ into an \emph{editable} blockchain protocol $\Blk'$.

%\TODOA{Need to say somewhere that we do one edit per block and it is easy to extend to more }
\subsection{Blockchain Protocol}\label{sec:abstract-new}

We consider an immutable blockchain protocol (for instance~\cite{EC:GarKiaLeo15}), denoted by $\Blk$, where nodes receive inputs from the environment $\env$, and interact among each other to agree on an ordered ledger that achieves persistence and liveness. The blockchain protocol $\Blk$ is characterised by a set of global parameters and by a public set of rules for validation. 
The protocol $\Blk$ provides the nodes with the following set of interfaces which are assumed to have complete access to the network and its users.

\begin{itemize}
\item $\{\chain',\bot\} \gets \Blk.\getChain$: returns a longer and \emph{valid} chain $\chain$ in the network (if it exists), otherwise returns $\bot$.

\item $\{0,1\} \gets \Blk.\chainValid(\chain)$: The chain validity check takes as input a chain $\chain$ and returns $1$ iff the chain is valid according to a public set of rules. 

\item $\{0,1\} \gets \Blk.\blockValid(B)$: The block validity check takes as input a block $B$ and returns $1$ iff the block is valid according to a public set of rules. 

%\item $\Blk.\mkTx(\txType,\data)$: takes as input the transaction type information and the transaction data. It then constructs a transaction of type $\txType$ with data $\data$, validate the transaction and include it in the next block.

%%\item $\bit \gets \Blk.\isStbl(\chain, \data)$: takes as input a chain $\chain$ and some transaction data $\data$ and checks if the transaction containing $\data$ is \emph{stabilised} (w.r.t. the persistence property) in $\chain$. If yes, then it returns $1$, otherwise it returns $0$.

\item $\Blk.\broadcast(\data)$: takes as input some data $\data$ and broadcasts it to all the nodes of the system.
	
\end{itemize}

\noindent The nodes in the $\Blk$ protocol have their own local chain $\chain$ which is initialised with a common genesis block. The consensus in $\Blk$ guarantees the properties of persistence and liveness discussed in~\cref{sec:blockchain_properties}.

\subsection{Editable Blockchain}

We build our editable blockchain protocol $\Blk'$ by modifying and extending the aforementioned protocol $\Blk$. The protocol $\Blk'$ has copies of all the basic blockchain functionalities exposed by $\Blk$ through the interfaces described above, and modifies the $\chainValid$ and $\blockValid$ algorithms in order to accommodate for edits in $\chain$. In addition, the protocol $\Blk'$ provides the following interfaces:
 \begin{itemize}
	\item $\candidateblk_j \leftarrow \Blk'.\redactreq(\chain,j,\data^\star)$: takes as input the chain $\chain$, an index $j$ of a block to edit and some data $\data^\star$. It then returns a candidate block for $B_j$.
	%\item \textcolor{blue}{$H(\ctr,G(s,x),y) \gets \Blk'.\votefunc(B)$: takes as input a block $B$ and parses $B$ as $(s,x,\ctr,y)$. It computes the hash value $H(\ctr,G(s,x),y)$ as a vote for $B$ and outputs the same.}
	\item $\bit \leftarrow \Blk'.\decide(\candidateblk_j, \chain)$: takes as input a candidate block $\candidateblk_j$ and the chain $\chain$ and returns $1$ iff the candidate block $\candidateblk_j$ is valid.
 \end{itemize}

The modified chain validation and block validation algorithms are presented in~\cref{alg:validate} and~\cref{alg:validateblk}, respectively, while the new algorithms to propose an edit to a  block and to validate candidate blocks are presented in~\cref{alg:request} and~\cref{alg:candidate}, respectively. In \cref{fig:protocol} we formally describe the protocol $\Blk'$.

Intuitively, we need modifications for chain validation and block validation algorithms to account for an edited block in the chain. A block that has been edited possesses a different state, that does not immediately correlate with its neighbouring blocks. Therefore, for such an edited block we need to ensure that the old state of the block (the state before the edit) is still accessible for verification.\footnote{Note that the protocol does \emph{not} need to maintain the redacted data for verification, and therefore all redacted data is completely removed from the chain.} We do this by storing the old state information in the block itself. This therefore requires a modified block validation algorithm and a modified chain validation algorithm overall.

We note that for simplicity our protocol is restricted to perform a single edit operation per block throughout the run of the protocol. In~\cref{apx:extension} we describe an extension of the protocol to accommodate for an arbitrary number of redactions per block.
%The protocol is, however easily extendable to many edits per block if the block instead records all the state changes that ever occurred within that block's lifetime.

\paragraph{Blockchain Policy}
We introduce the notion of a blockchain policy $\policy$, that determines if an edit to the chain $\chain$ should be approved or not. The protocol $\Blk'$ is parameterised by a policy $\policy$ that is a function that takes as input a chain $\chain$ and a candidate block $\candidateblk$ (that proposes a modification to the chain $\chain$) and it returns $\accept$ if the candidate block $\candidateblk$ complies with the policy $\policy$, otherwise it outputs $\reject$; in case the modification proposed by $\candidateblk$ is still being deliberated in the chain $\chain$, then $\policy$ returns $\voting$. 

In its most basic form, a policy $\policy$ requires that a candidate block $\candidateblk$ should only be accepted if $\candidateblk$ was voted by the majority of the network within some predefined interval of blocks (or \emph{voting period} $\ell$). A formal definition follows.

\begin{definition}[Policy]\label{def:policy}
A candidate block $\editreq$ generated in round $r$ is said to satisfy the policy $\policy$ of chain $\chain := (B_1,\ldots,B_n)$, i.e., $\policy(\chain, \editreq) = \accept$, if it holds that $B_{r+\ell} \in \prune{\chain}{k}$ and the ratio of blocks between $B_r$ and $B_{r+\ell}$ containing $H(\editreq)$ (a vote for $\editreq$) is at least $\rho$, for $k,\ell \in \NN$, and $0 < \rho \le 1$, where $k$ is the persistence parameter, $\ell$ is the voting period, and $\rho$ is the ratio of votes necessary within the voting period $\ell$.
\end{definition}

\restylefloat{figure}
\floatstyle{boxed}

\begin{figure*}
The protocol $\Blk'$ consists of a sequence of rounds $\round$, and is parameterised by the liveness and persistence parameters, denoted by $u,k$, respectively, and by a policy $\policy$ that among other rules and constraints, determines the parameter $\ell$ (that is the duration of the voting period) and $\rho$ (that is the threshold of votes within the period $\ell$ for a candidate block to be accepted and incorporated into the chain). A pictorial representation of the protocol can be found in~\cref{fig:abstract_protocol}.
\begin{flushleft}
\textbf{Initialisation.} Set the chain $\chain \leftarrow \genesis$, set round $\round \leftarrow 1$ and initialise an empty list of candidate blocks for edits $\rcbpool := \emptyset$.
\end{flushleft}
For each round $\round$ of the protocol, we describe the following sequence of execution.
\begin{flushleft}
\textbf{Chain update.} At the beginning of a new round $\round$, the nodes try to update their local chain by calling $\chain \leftarrow \Blk'.\getChain$. 	
\end{flushleft}

\begin{flushleft}
\textbf{Candidate blocks pool.} Collect all candidate blocks $\candidateblk_j$ from the network and add $\candidateblk_j$ to the pool of candidate blocks $\rcbpool$ iff $\Blk'.\decide(\chain,\candidateblk_j) = 1$; otherwise discard $\candidateblk_j$.
\end{flushleft}

\begin{flushleft}		
\textbf{Editing the chain.} For all candidate blocks $\candidateblk_j \in \rcbpool$ do:
	\begin{itemize}
		\item If $\policy(\chain,\candidateblk_j) = \accept$, then build the new chain as $\chain \leftarrow \prune{\chain}{(n-j+1)}||\candidateblk_j||\pruneback{\chain}{j}$ and remove $\candidateblk_j$ from $\rcbpool$. For policy $\policy$ to accept $\candidateblk_j$, it must be the case that the ratio of votes for $\candidateblk_j$ within its voting period ($\ell$ blocks) is at least $\rho$.
		\item If $\policy(\chain,\candidateblk_j) = \reject$, then remove $\candidateblk_j$ from $\rcbpool$. For policy $\policy$ to reject $\candidateblk_j$ it must be the case that the ratio of votes for $\candidateblk_j$ within its voting period ($\ell$ blocks) is less than $\rho$.
		\item If $\policy(\chain,\candidateblk_j) = \voting$, then do nothing. 
	\end{itemize}
\end{flushleft}
	
%\begin{flushleft}
%\textbf{Voting for a redaction.} For all candidate blocks $\candidateblk_j \in \rcbpool$ that the node is willing to support, if $\policy(\candidateblk_j,\chain) = \voting$ then call $\Blk'.\mkTx(\voteTx, H(\candidateblk_j))$; otherwise do nothing.
%\end{flushleft}

\begin{flushleft}
\textbf{Creating a new block.} Collects all the transaction data $x$ from the network for the $\round$-th round and tries to build a new block $B_\round$ by performing the following steps: 
			\begin{itemize}
 				\item \emph{(Voting for candidate blocks).} For all candidate blocks $\candidateblk_j \in \rcbpool$ that the node is willing to endorse, if $\policy(\chain,\candidateblk_j) = \voting$ then set $x \leftarrow x || H(\candidateblk_j)$.
				\item Create a new block $B := \langle s, x, \ctr, G(s,x)\rangle$, such that $s = H(\ctr',G(s',x'),y')$, for $\langle s',x',\ctr',y'\rangle \leftarrow \head{\chain}$.
				\item Extend its local chain $\chain \leftarrow \chain \vert\vert B$ and iff $\Blk'.\chainValid(\chain) = 1$ then broadcast $\chain$ to the network.
			\end{itemize}
\end{flushleft}

\begin{flushleft}
\textbf{Propose an edit.} The node willing to propose an edit for the block $B_j$, for $j \in [n]$, creates a candidate block $\candidateblk_j \leftarrow \Blk'.\redactreq(\chain,j, \data^\star)$ using the new data $\data^\star$, and broadcasts it to the network by calling $\Blk'.\broadcast(\candidateblk_j)$. 
\end{flushleft}

\caption{Accountable permissionless editable blockchain protocol $\Blk'_\policy$}
\label{fig:protocol}
\end{figure*}

\subsection{Protocol Description}\label{sec:protocol}
We denote a block to be of the form $B := \blockmod$, where $s \in \bit^\spar$ is the hash of the previous block, $x \in \bit^*$ is the block data, and $y \in \{ 0,1\}^\spar$ is the old state of the block data. 
To extend an editable chain $\chain$ to a new longer chain $\chain' := \chain||B'$, the newly created block $B' := \langle s',x',\ctr', y'\rangle$ sets $s' := H(\ctr,G(s,x),y)$, where $\head\chain := \blockmod$. Note that upon the creation of block $B'$, the component $y'$ takes the value $G(s',x')$, that represents the initial state of block $B'$. 

During the setup of the system, the chain $\chain$ is initialised as $\chain:=\genesis$, and all the users in the system maintain a local copy of the chain $\chain$ and a pool $\rcbpool$ consisting candidate blocks for edits, that is initially empty. The protocol runs in a sequence of rounds $\round$ (starting with $\round := 1$).

%\floatstyle{plain}
%\restylefloat{figure}
%\begin{figure}
%\centering
%\begin{tabular}{p{.40\columnwidth} p{.46\columnwidth}}
%\textbf{Value} & \textbf{Description} \\ \thickhline
%$\mathtt{state\_link}\ (s)$ & link to the previous block\\ \hline
%$\mathtt{data}\ (x)$ & data items stored in the block\\ \hline
%$\mathtt{counter}\ (\ctr)$ & proof of mining the block correctly\\ \thickhline
%%$\mathtt{timestamp}$ & the timestamp of the block \\ \hline
%%%$\mathtt{nonce}$ & nonce used in proof-of-work \\ \thickhline
%\rowcolor{Gainsboro!30}$\mathtt{old\_state}\ (y)$ & original state of the block when created $(=G(s,x))$\\ \thickhline
%
%\end{tabular}%}
%\caption{Structure of the block $B$. The last highlighted field ($\mathtt{old\_state}$) is only included in the block of the extended (editable) protocol. A candidate block $B^\star$ for $B$ is of the form $(s,x^\star,\ctr,y)$ where only the data stored has been edited.}
%\label{fig:modifyblockabstract}
%\end{figure}

In the beginning of each round $\round$, the users try to extend their local chain using the interface $\Blk'.\getChain$, that tries to retrieve new valid blocks from the network and append them to the local chain. Next, the users collect all the candidate blocks $\candidateblk_j$ from the network and validate them by using $\Blk'.\decide$ (\cref{alg:candidate}); then, the users add all the valid candidate blocks to the pool $\rcbpool$. 
For each candidate block $\candidateblk_j$ in $\rcbpool$, the users compute $\policy(\chain,\candidateblk_j)$ to verify if the candidate block $\candidateblk_j$ should be adopted by the chain or not; if the output is $\accept$ they replace the original block $B_j$ in the chain by the candidate block $\candidateblk_j$ and remove $\candidateblk_j$ from $\rcbpool$. If the output is $\reject$, the users remove the candidate block $\candidateblk_j$ from $\rcbpool$, otherwise if the output is $\voting$ they do nothing. 
To create a new block $B$ the users collect transactions from the network and store them in $x$; if a user wishes to endorse the edit proposed by a candidate block $\candidateblk_j \in \rcbpool$ that is still in $\voting$ stage, the user can vote for the candidate block $\candidateblk_j$ by simply adding $H(\candidateblk_j)$ to the data $x$. After the block is created and the new extended chain $\chain':= \chain||B$ is built, the users broadcast the new chain $\chain'$ iff $\Blk'.\chainValid(\chain') = 1$ (\cref{alg:validate}). 
Finally, if a user wishes to propose an edit to block $B_j$ in the chain $\chain$, she first creates the new data $x^\star_j$, that represents the modifications that she proposes to make to the data $x_j$, and calls $\redactreq$ (\cref{alg:request}) using the interface $\Blk'.\redactreq$ with the chain $\chain$, index $j$ of the block in $\chain$ and the new data $x^\star_j$. The algorithm returns a candidate block $\candidateblk_j$ that is broadcasted to the network.

%%%%%%%%% Validate chain
%%%%%%%%%%%%%%%%%%%%%%%%%%%%%%%%%%%%%%%%%%%%%%%%%%%%%%%

\paragraph{Chain Validation} Given a chain $\chain$, the user needs to validate $\chain$ according to some set of validation rules. To do this, she uses the $\Blk'.\chainValid$ interface, that is implemented by~\cref{alg:validate}. The algorithm takes as input a chain $\chain$ and starts validating from the head of $\chain$. In~\Cref{line:block_valid}, the validity of the block $B_j$ is checked. If the assertion in~\cref{line:new_state} is false and if the check in~\cref{line:old_state} is successful, then the block $B_{j-1}$ is a valid edited block. In~\cref{line:old_state}, the validity of $B_{j-1}$ is checked in the context of a candidate block and whether the block is accepted according to the voting policy $\policy$ of the chain.  

\begin{algorithm}
\SetAlgoVlined
\SetKwInOut{Input}{input}\SetKwInOut{Output}{output}
\Input{Chain $\chain = (B_1,\cdots,B_n)$ of length $n$.}
\Output{$\bit$}

\BlankLine
$j:=n$\;
\lIf{$j = 1$}{\Return $\Blk'.\blockValid(B_1)$}
\While{$j \ge 2$}{ 
$B_j:= \langle s_j,x_j,\ctr_j,y_j\rangle$ \Comment*{\scriptsize{$B_j := \head{\chain}$ when $j=n$}}\
\lIf{$\Blk'.\blockValid(B_j) = 0$}{\Return $0$}\label{line:block_valid} %remember to validate also the contents 
\lIf{$s_j = H(\ctr_{j-1}, G(s_{j-1},x_{j-1}),y_{j-1})$}{$j := j-1$}\label{line:new_state}
\lElseIf{$(s_j = H(\ctr_{j-1}, y_{j-1},y_{j-1})) \land (\Blk'.\decide(\chain,B_{j-1}) = 1) \land (\policy(\chain,B_{j-1}) =\accept)$}{$j := j-1$}\label{line:old_state}
\lElse{\Return $0$}
}
\Return $1$\;
\caption{$\chainValid$ \footnotesize{(implements $\Blk'.\chainValid$)}}	
\label{alg:validate}
\end{algorithm}

%%%%%%%%% Validate block
%%%%%%%%%%%%%%%%%%%%%%%%%%%%%%%%%%%%%%%%%%%%%%%%%%%%%%%

\paragraph{Block Validation} To validate a block, the $\blockValid$ algorithm (described in~\cref{alg:validateblk}) takes as input a block $B$ and first validates the data included in the block according to some pre-defined validation predicate. It then checks if the block indeed satisfies the constraints of the PoW puzzle. Apart from this check, the or ($\lor$) condition is to ensure that in case of dealing with an edited block $B$, the old state of $B$ still satisfies the PoW constraints.    

\begin{algorithm}
\SetAlgoVlined
\SetKwInOut{Input}{input}\SetKwInOut{Output}{output}
\Input{Block $B:=\blockmod$.}
\Output{$\bit$}

\BlankLine
Validate data $x$, if \emph{invalid} \Return $0$\;
\lIf{$H(\ctr,G(s,x),y) < D \lor H(\ctr, y,y) < D$}{\Return $1$}\label{line:pow}

\lElse{\Return $0$}\label{line:ret}
\caption{$\blockValid$ \footnotesize{(implements $\Blk'.\blockValid$)}}	
\label{alg:validateblk}
\end{algorithm} 

%%%%%%% Requesting
%%%%%%%%%%%%%%%%%%%%%%%%%%%%%%%%%%%%%%%%%%%%%%%%%%%%%%

\paragraph{Proposing an Edit} Any user in the network can propose for a particular data to be removed or replaced from the blockchain. She uses the $\redactreq$ algorithm as described in~\cref{alg:request} and constructs a candidate block to replace the original block. The algorithm takes as input a chain $\chain$, the index $j$ of the original block and new data $x_j^\star$ that will replace the original data. If the user's intention is simply to remove all data from block $B_j$ then $x_j^\star := \varepsilon$. It then generates a candidate block as the tuple $B_j^\star := \langle s_j,x_j^\star,\ctr_j,y_j\rangle$.

\begin{algorithm}
  \SetKwInOut{Input}{input}\SetKwInOut{Output}{output}
   \Input{Chain $\chain = (B_1,\cdots,B_n)$ of length $n$, an index $j \in [n]$, and the new data $x_j^\star$.}
  \Output{A candidate block $\candidateblk_j$.}
  \BlankLine 
    Parse $B_j:= \blockmodi{j}$\;
    Build the candidate block $\candidateblk_j := \langle s_j,x_j^\star,\ctr_j,y_j\rangle$\;
    \Return $\candidateblk_j$\;
 \caption{$\redactreq$ \footnotesize{(implements $\Blk'.\redactreq$)}}
 \label{alg:request}
\end{algorithm}

\paragraph{Validating Candidate Blocks} When the user wishes to validate a candidate block $\candidateblk_j := \langle s_j,x^\star_j,\ctr_j,y_j\rangle$ for the $j$-th block of a chain $\chain$, she uses $\decide$ which is described in~\cref{alg:candidate}. It retrieves the blocks $B_{j-1}$ and $B_{j+1}$ of index $j-1$ and $j+1$ respectively from the chain $\chain$. In~\cref{line:links} it is checked if the link $s^\star_j$ from $\candidateblk_j$ to $B_{j-1}$ holds and that the link $s_{j+1}$ from $B_{j+1}$ to $B_{j}^\star$ also satisfies the condition $s_{j+1} = H(ctr_j,y_j,y_j)$. The latter condition checks if the ``old link" still holds. If both checks are successful the candidate block $\candidateblk_j$ is considered valid, otherwise it is considered invalid. 

\begin{algorithm}
	\SetKwInOut{Input}{input}\SetKwInOut{Output}{output}
\Input{Chain $\chain = (B_1,\cdots,B_n)$ of length $n$, and a  candidate block $\candidateblk_j$ for an edit.}
\Output{$\bit$}
\BlankLine
Parse $\candidateblk_j := \langle s_j,x^\star_j, \ctr_j, y_j \rangle$\;
\lIf{$\Blk'.\blockValid(\candidateblk_j) = 0$}{\Return $0$}
Parse $B_{j-1}:= \blockmodi{j-1}$\;
Parse $B_{j+1}:= \blockmodi{j+1}$\;
\lIf{$s^\star_j = H(\ctr_{j-1},y_{j-1},y_{j-1}) \land s_{j+1} = H(\ctr_j,y_j,y_j)$}{\Return $1$}\label{line:links}
\lElse{\Return $0$}

\caption{$\decide$ \footnotesize{(implements $\Blk'.\decide$)}}
\label{alg:candidate}
\end{algorithm}

%% file: analysis.tex
\section{Security Analysis}\label{sec:security}
In this section we analyse the security of our editable blockchain protocol of \cref{fig:protocol}.

We assume the existence of an immutable blockchain protocol $\Blk$, as described in~\cref{sec:abstract-new}, that satisfies the properties of chain growth, chain quality and common prefix~\cite{EC:GarKiaLeo15}. The basic intuition behind our security analysis is that, given that $\Blk$ satisfies the aforementioned properties, our editable blockchain protocol $\Blk'_\policy$, (which is $\Blk'$ parameterised by a policy $\policy$), preserves the same properties (or a variation of the property in the case of common prefix). Therefore, our protocol behaves exactly like the immutable blockchain $\Blk$ when there are no edits in the chain, and if an edit operation was performed, it must have been approved by the policy $\policy$. We discuss each individual property next.

\paragraph{Chain Growth} The chain growth property from $\Blk$ is automatically preserved in our editable blockchain $\Blk'$, since the possible edits do not allow the removal of blocks or influence the growth of the chain. We present the formal definition next, followed by a theorem stating that $\Blk'$ preserves chain growth whenever $\Blk$ satisfies chain growth. 

\begin{definition}[Chain Growth~\cite{EC:GarKiaLeo15}]
Consider the chains $\chain_1,\chain_2$ possessed by two honest parties at the onset of two slots $\slot_1,\slot_2$, with $\slot_2$ at least $s$ slots ahead of $\slot_1$. Then it holds that $\length{\chain_2}-\length{\chain_1} \ge \tau \cdot s$, for $s \in \NN$ and $0<\tau \le 1$, where $\tau$ is the speed coefficient.
\end{definition}

\begin{theorem}
If $\Blk$ satisfies $(\tau,s)$-chain growth, then $\Blk'_\policy$ satisfies $(\tau,s)$-chain growth for any policy $\policy$.
\end{theorem}

\begin{proof}	
We note that $\Blk'$ extends $\Blk$, that by assumption satisfies chain growth. Also, note that in $\Blk'$ it is not possible to remove a block from the chain (for any policy $\policy$), thereby reducing the length of $\chain$. In other words, the edits performed do not alter the length of the chain. Therefore, we conclude that $\Blk'$ satisfies chain growth whenever $\Blk$ satisfies chain growth.

%Since the editable blockchain protocol $\Blk'$ does not have any other differences to the mining or consensus process of $\Blk$, it is easy to see that $\Blk'$ inherits chain growth from $\Blk$.
\end{proof}

\paragraph{Chain Quality} The chain quality property informally states that the ratio of adversarial blocks in any segment of a chain held by a honest party is no more than a fraction $\mu$, where $\mu$ is the fraction of resources controlled by the adversary.

\begin{definition}[Chain Quality~\cite{EC:GarKiaLeo15}]
Consider a portion of length $\ell$-blocks of a chain possessed by an honest party during any given round, for $\ell \in \NN$. Then, the ratio of adversarial blocks in this $\ell$ segment of the chain is at most $\mu$, where $0 <\mu \le 1$ is the chain quality coefficient.		
\end{definition}

\begin{theorem}
Let $H$ be a collision-resistant hash function. If $\Blk$ satisfies $(\mu,\ell)$-chain quality, then $\Blk'_\policy$ satisfies $(\mu,\ell)$-chain quality for any $(k,\ell,\rho)$-policy where $\rho > \mu$.  	
\end{theorem}

\begin{proof}
%Towards contradiction, assume that there exists an adversary $\advA'$ that breaks the chain quality property of $\Blk'_\policy$. Then, we build another adversary $\advA$ that breaks the chain quality of $\Blk$.
We note that the only difference in $\Blk'_\policy$ in relation to $\Blk$ is that blocks can be edited. An adversary $\advA$ could edit an honest block $B$ in the chain $\chain$ into a malicious block $\candidateblk$ (e.g.,\ that contains illegal content), increasing the proportion of malicious blocks in the chain, and therefore breaking the chain quality property. We show below that $\advA$ has only a negligible probability of violating chain quality of $\Blk'$.

Let $\advA$ propose a malicious candidate block $B_j^\star$ for editing an honest block $B_j \in \chain$. Since $\advA$ possesses only $\mu$ computational power, by the chain quality property of $\Blk$ we know that the adversary mines at most $\mu$ ratio of blocks in the voting phase. As the policy stipulates, the ratio of votes has to be at least $\rho$ for $B^\star$ to be approved, where $\rho > \mu$. Therefore, $B^\star$ can only be approved by the policy $\policy$ if honest nodes vote for it. Observe that the adversary could try to build an ``honest looking" (e.g.,\ without illegal contents) candidate block $\widetilde{\candidateblk} \ne \candidateblk$ such that $H(\widetilde{\candidateblk}) = H(\candidateblk)$, in an attempt to deceive the honest nodes during the voting phase; the honest nodes could endorse the candidate block $\candidateblk$ during the voting phase, and the adversary would instead edit the chain with the malicious block $\widetilde{\candidateblk}$. The adversary has only a negligible chance of producing such a candidate block $\candidateblk$ where $H(\widetilde{\candidateblk}) = H(\candidateblk)$, since this would violate the collision-resistance property of the hash function $H$.

Moreover, $B^\star$ is incorporated to the chain only if it is an honest candidate block. This concludes the proof.

% Observe that the adversary will violate collision resistance of $H$ if he produces $B^{\star\star} \ne B^\star$ such that they both have the same vote $H(B^{\star\star}) = H(B^\star)$.
%For the redaction $\candidateblk_j$ to be accepted by policy $\policy$, it needs to get a proportion of votes bigger than $\mu$ during the voting phase. 
\end{proof}

\paragraph{Common Prefix} The common prefix property informally says that if we take the chains of two honest nodes at different time slots, the shortest chain is a prefix of the longest chain (up to the common prefix parameter $k$). We show the formal definition next.

\begin{definition}[Common Prefix~\cite{EC:GarKiaLeo15}]\label{def:common_prefix}
The chains $\chain_1,\chain_2$ possessed by two honest parties at the onset of the slots $\slot_1 < \slot_2$ are such that $\prune{\chain_1}{k} \preceq \chain_2$, where $\prune{\chain_1}{k}$ denotes the chain obtained by removing the last $k$ blocks from $\chain_1$, where $k \in \NN$ is the common prefix parameter.
\end{definition}

We remark however, that our protocol $\Blk'_\policy$ inherently does not satisfy~\cref{def:common_prefix}. To see this, consider the case where two chains $\chain_1$ and $\chain_2$ are held by two honest parties $P_1$ and $P_2$ at slots $\slot_1$ and $\slot_2$ respectively, such that $\slot_1 < \slot_2$. In slot $r$ starts the voting phase (that lasts $\ell$ blocks) for a candidate block $\candidateblk_j$ proposing to edit block $B_j$, such that $j + k \le r < \slot_1 \le \ell +k < \slot_2$. Note that at round $\slot_1$ the voting phase is still on, therefore $\policy(\chain_1,\candidateblk_j) = \voting$. By round $\slot_2$, the voting phase is complete and in case $\policy(\chain_2,\candidateblk_j) = \accept$ the block $B_j$ is replaced by $\candidateblk_j$ in $\chain_2$. However, in chain $\prune{\chain_1}{k}$ the $j$-th block is still $B_j$, since the edit of $\candidateblk_j$ is waiting to be confirmed. Therefore, $\prune{\chain_1}{k} \nprec \chain_2$, thereby violating~\cref{def:common_prefix}. 

The pitfall in~\cref{def:common_prefix} is that it does not account for edits or modifications in the chain. We therefore introduce a new definition that is suited for an editable blockchain (with respect to an editing policy).  The formal definition follows.

\begin{definition}[Editable Common prefix]\label{def:editable_common_prefix}
The chains $\chain_1,\chain_2$ of length $l_1$ and $l_2$, respectively, possessed by two honest parties at the onset of the slots $\slot_1 \le \slot_2$ satisfy one of the following:
	\begin{enumerate}
		\item $\prune{\chain_1}{k} \preceq \chain_2$, or
		\item for each $\candidateblk_j \in \prune{\chain_2}{(l_2-l_1)+k}$ such that $\candidateblk_j \notin \prune{\chain_1}{k}$, it must be the case that $\policy(\chain_2,\candidateblk_j) = \accept$, for $j\in [l_1-k]$,
	\end{enumerate}
	where $\prune{\chain_2}{(l_2-l_1)+k}$ denotes the chain obtained by pruning the last $(l_2-l_1)+k$ blocks from $\chain_2$, $\policy$ denotes the chain policy, and $k \in \NN$ denotes the common prefix parameter.
\end{definition}  

Intuitively, the above definition states that if there exists a block that violates the common prefix as defined in~\cref{def:common_prefix}, then it must be the case that this block is an edited block whose adoption was voted and approved according to the policy $\policy$ in chain $\chain_2$. We show that our protocol $\Blk'$ satisfies~\cref{def:editable_common_prefix} next.

\begin{theorem}\label{thm:editable_common_prefix}
Let $H$ be a collision-resistant hash function. If $\Blk$ satisfies $k$-common prefix, then $\Blk'_\policy$ satisfies $k$-editable common prefix for a $(k,\ell,\rho)$-policy.
\end{theorem}

\begin{proof}
If no edits were performed in a chain $\chain$, then the protocol $\Blk'_\policy$ behaves exactly like the immutable protocol $\Blk$, and henceforth the common prefix property follows directly.
	
However, in case of an edit, consider an adversary $\advA$ that proposes a candidate block $\candidateblk_j$ to edit $B_j$ in chain $\chain_2$, which is later edited by an honest party $P_2$ at slot $\slot_2$. Observe that by the collision resistance property of $H$, $\advA$ is not able to efficiently produce another candidate block $\widetilde{\candidateblk_j} \ne \candidateblk_j$ such that $H(\widetilde{\candidateblk_j}) = H(\candidateblk_j)$. Therefore, since $P_2$ is honest and adopted the edit $\candidateblk_j$ in $\chain_2$, it must be the case that $\candidateblk_j$ received enough votes such that $\policy(\chain_2,\candidateblk_j) = \accept$. This concludes the proof.
\end{proof}

\noindent\emph{How the properties play together}: By showing that $\Blk'$ satisfies the three aforementioned properties, we show that $\Blk'_\policy$ is a live and persistent blockchain protocol immutable against edits not authorised by the policy $\policy$.

The editable common prefix property ensures that only policy approved edits are performed on the chain. The Chain quality property, for a $(k,\ell,\rho)$-policy $\policy$ where $\rho > \mu$, ensures that an adversary does not get a disproportionate contribution of blocks to the chain.
% By having $\rho > \mu$ we incorporate only honest edits into the chain. 

%% file: bitcoin.tex
\section{Integrating into Bitcoin}\label{sec:instantiation}

In this section we describe how our generic editable blockchain protocol (\cref{fig:protocol}) can be integrated into Bitcoin. For simplicity, we consider one redaction per block and the redaction is performed on one or more transactions included in the block. The extension of the generic protocol for multiple redactions (described in~\cref{apx:extension}) can be immediately applied to the construction described in this section. Next, we give a brief background on the Bitcoin protocol.

\subsection{Bitcoin Basics}
\paragraph{Transactions} A simple transaction $\tx$ in Bitcoin has the following basic structure: an input script, an output script with a corresponding amount, and a witness. More complex transactions may have multiple input and output scripts and/or more complex scripts. A transaction $\tx'$ that spends some output $\tau$ of $\tx$, has the ID of $\tx$ in its input, denoted by $\tx_\id := H(\tx)$, and a witness $x$ that satisfies the output script $\tau$ of $\tx$ (as shown in~\cref{fig:transaction}). The amount $\amount_2$ being spent by the output script $\tau_2$ needs to be smaller (or equal) than the amount $\amount_1$ of $\tau_1$. The most common output scripts in Bitcoin consists of a public key, and the witness $x$ is a signature of the transaction computed using the corresponding secret key. We refer the reader to~\cite{script} for a comprehensive overview of the Bitcoin scripting language.

\paragraph{Insertion of Data} Users are allowed to propose new transactions containing arbitrary data, that are then sent to the Bitcoin network for a small fee. Data can be inserted into specific parts of a Bitcoin transaction, namely the output script, input script and witness. Matzutt et al.~\cite{matzutt2018quantitative} provide a quantitative analysis of data insertion methods in Bitcoin. According to their analysis, $\mathtt{OP\_RETURN}$ and coinbase transactions are the major pockets apart from some non-standard transactions, where data is inserted.

%%%%%%%%%%%%%% Bitcoin transactions
%%%%%%%%%%%%%%%%%%%%%%%%%%%%%%%%%%%%%%%%%%%%%%%%%%%%%
\floatstyle{plain}
\restylefloat{figure}
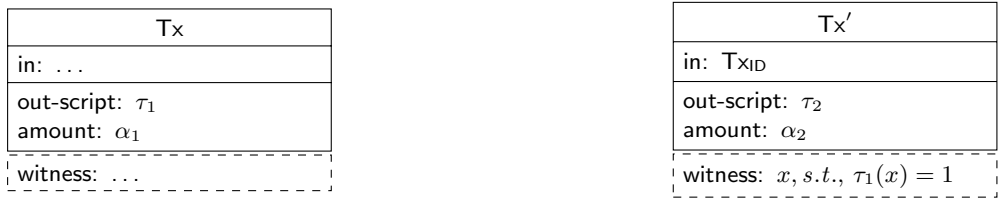
\begin{figure}
\begin{minipage}{0.45\linewidth}
\begin{center}

\begin{tikzpicture}
     \node [box=3]
    (BE)
    {
    \nodepart[align=center]{one}
    {$\tx$}
    \nodepart{two} in: \dots
      \nodepart{three} out-script: $\tau_1$\\
      amount: $\amount_1$
    };
    \node [below=of BE, wbox]
    (WE)
    { 
    witness: \dots
    };
    \end{tikzpicture}

\end{center}
\end{minipage}
\hfill
\begin{minipage}{0.45\linewidth}
\begin{center}

\begin{tikzpicture}
  \node [box=3]
    (BE)
    {
    \nodepart[align=center]{one}
    {$\tx'$}
    \nodepart{two} in: $\tx_\id$
      \nodepart{three} out-script: $\tau_2$\\
      amount: $\amount_2$
    };
    \node [below=of BE, wbox]
    (WE)
    { 
    witness: $x, \mathit{s.t.,}\ \tau_1(x) = 1$
    };
    \end{tikzpicture}
\end{center}
\end{minipage}
\caption{The structure of a transaction in Bitcoin. The transaction $\tx'$ is spending the output $\tau_1$ of transaction $\tx$.}
\label{fig:transaction}
\end{figure}

\paragraph{Block Structure} A Bitcoin block consists of two parts, namely the block header, and a list of all transactions within the block. The structure of the block header is detailed in~\cref{fig:modifyheader}, whereas a pictorial representation of the list of transactions can be found in~\cref{fig:modifytxlist}.

%%%%%%%% Bitcoin block header
%%%%%%%%%%%%%%%%%%%%%%%%%%%%%%%%%%%%%%%
%\begin{figure}
%\centering
%\begin{tabularx}{\columnwidth}{>{\hsize=.5\hsize}XX}
%\textbf{Value} & \textbf{Description} \\ \thickhline
%$\mathtt{hash\_prev}$ & hash of the previous block header \\ \hline
%$\mathtt{merkle\_root}$ & root of the merkle tree (whose the leaves are the transactions)\\ \hline
%$\mathtt{difficulty}$ & the difficulty of the proof-of-work\\ \hline
%$\mathtt{timestamp}$ & the timestamp of the block \\ \hline
%$\mathtt{nonce}$ & nonce used in proof-of-work \\ \thickhline
%\end{tabularx}
%\caption{Contents of Bitcoin block header}
%\label{fig:header}
%\end{figure}

\subsection{Modifying the Bitcoin Protocol}
In this section we detail the modifications to the Bitcoin protocol necessary to integrate it to our generic editable blockchain protocol of~\cref{sec:editing}. The resulting protocol is a version of Bitcoin that allows for redaction of (harmful) data from its transactions.

By redaction of transactions, we mean removing data from a transaction without making other changes to the remaining components of the transaction. As shown in~\cref{fig:data_insert_removal_a}, consider a transaction $\tx_1$ that contains some harmful data in its output script, and let $\tx_1^\star$ be a candidate transaction to replace $\tx$ in the chain, where $\tx_1^\star$ is exactly the same as $\tx_1$, except that the harmful data is removed (~\cref{fig:data_insert_removal_b}).

%%%%%%%%%%%%%%%%%% Edit Tx
%%%%%%%%%%%%%%%%%%%%%%%%%%%%%%%%%%%

\begin{figure}[h]
\centering
\begin{tikzpicture}

  \node [box=3]
    (BE)
    {
    \nodepart[align=center]{one}
    {$\redactTx$}
    \nodepart{two} in: \dots
      \nodepart{three} out-script: $\tx_{1_\id},{\tx_{1_\id}^\star}$
     };
    \node [below=of BE,wbox]
    (WE)
    { 
    witness: \dots
    };
    \end{tikzpicture}
    
\caption{The special transaction $\redactTx$ is broadcasted to the network to propose a redaction of transaction $\tx_1$ for the candidate transaction $\tx_1^\star$.}
\label{fig:editrequest}
\end{figure}
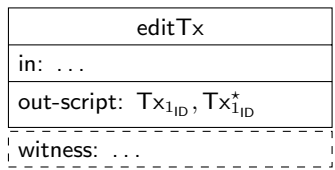

\paragraph{Proposing Redactions} A user who wishes to propose a redaction proceeds as follows: First, constructs a special transaction $\redactTx$ (as shown in~\cref{fig:editrequest}) containing $\tx_{1_\id}$ and $\tx_{1_\id}^\star$, that respectively denotes the hash of the transaction $\tx_1$ being redacted, and the hash of $\tx_1^\star$ that is the candidate transaction to replace $\tx_1$ in the chain\footnote{We note that our transaction ID is Segwit compatible, as the witness is not used with the hash $H$ to generate a transaction's ID.}. Then, broadcasts the special transaction $\redactTx$ and the candidate transaction $\tx_{1}^\star$ to the network; $\redactTx$ requires a transaction fee to be included in the blockchain, while $\tx_{1}^\star$ is added to a pool of candidate transactions\footnote{If a candidate transaction does not have a corresponding $\redactTx$ in the blockchain then the transaction is not included in the candidate pool, and it is treated as spam instead.}. The candidate transaction $\tx_{1}^\star$ is validated by checking its contents with respect to $\tx_{1}$, and if it is valid, then it can be considered for voting.   

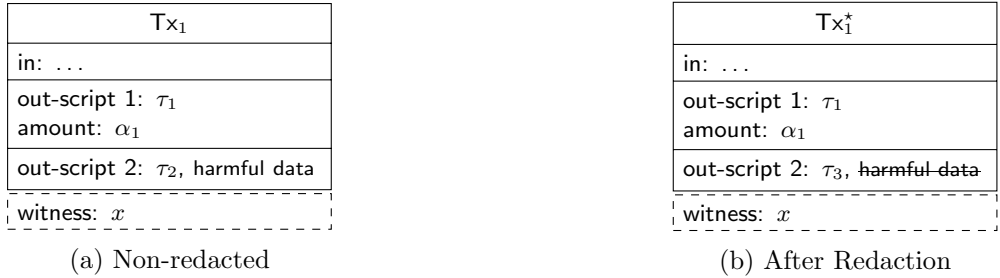
\begin{figure}[]
\begin{minipage}{0.45\linewidth}
\begin{center}

\begin{tikzpicture}
  \node [box=4]
    (BE)
    {
    \nodepart[align=center]{one}
    {$\tx_1$}
    \nodepart{two} in: \dots
      \nodepart{three} out-script 1: $\tau_1$\\
      amount: $\amount_1$
            \nodepart{four} out-script 2: $\tau_2$, \scriptsize{harmful data}
    };
    \node [below=of BE,wbox]
    (WE)
    { 
    witness: $x$
    };
    \end{tikzpicture}
\subcaption{Non-redacted}
\label{fig:data_insert_removal_a}
\end{center}
\end{minipage}
\hfill
\begin{minipage}{0.45\linewidth}
\begin{center}

\begin{tikzpicture}

  \node [box=4]
    (BE)
    {
    \nodepart[align=center]{one}
    {$\tx_1^\star$}
    \nodepart{two} in: \dots
      \nodepart{three} out-script 1: $\tau_1$\\
      amount: $\amount_1$
            \nodepart{four} out-script 2: $\tau_3$, \scriptsize{\st{harmful data}}
    };
    \node [below=of BE,wbox]
    (WE)
    { 
    witness: $x$
    };
    \end{tikzpicture}
    \subcaption{After Redaction}
        \label{fig:data_insert_removal_b}
\end{center}
\end{minipage}
\caption{(a) The transaction $\tx_1$ contains harmful data, and (b) the candidate transaction $\tx_1^\star$ contains a copy of all the fields of $\tx_1$, with exception of the harmful data.}
\label{fig:data_insert_removal}
\end{figure}

%A miner who wishes to propose an edit request does the following. Consider the case where the miner wants to edit $\tx_1$ and modify it into $\tx_{1}^\star$ (\cref{fig:data_insert_removal}). He constructs $\redactTx$ (\cref{fig:editrequest}), a special transaction which contains $\tx_{1_\id}$ and $\tx_{1_\id}^\star$. He then broadcasts a message $m$ consisting of $\redactTx$ and the candidate transaction $\tx_{1}^\star$ to the network. $\redactTx$ is included in the blockchain while $\tx_{1}^\star$ is added to the candidate pool. Using $\tx_{1_\id}$, the old transaction $\tx_{1}$ is identified from the blockchain. The candidate transaction $\tx_{1}^\star$ is now validated by checking its contents with respect to $\tx_{1}$ and is considered for voting by the miners. If a candidate transaction does not have a corresponding $\redactTx$ in the blockchain, then it is removed from the candidate pool as spam.  

\paragraph{Redaction Policy} 
The redactable Bitcoin protocol is parameterised by a policy parameter $\policy$ (\cref{def:policy}). The policy $\policy$ dictates the requirements and constraints for redaction operations in the blockchain. An informal description of a (basic) policy for Bitcoin would be:

 A proposed redaction is approved valid if the following conditions hold:
 \begin{itemize}
  	\item It is identical to the transaction being replaced, except that it can remove data.

 	\item It can only remove data that can never be spent, e.g.,\ $\mathtt{OP\_RETURN}$ output scripts.
     
    \item It does not redact votes for other redactions in the chain. 
 	
 	\item It received more than $50$\% of votes in the $1024$ consecutive blocks (voting period) after the corresponding $\redactTx$ is stable in the chain.

 \end{itemize}

\noindent where voting for a candidate transaction $\tx_{1}^\star$ simply means that the miner includes $\redactTx_\id = H(\tx_{1_\id} || \tx_{1_\id}^\star)$ in the coinbase (transaction) of the new block he produces. After the voting phase is over, the candidate transaction is removed from the candidate pool. 

The reason for restricting the redactions to non-spendable components of a transaction (e.g.,\ $\mathtt{OP\_RETURN}$) is that, permitting redactions on spendable content could lead to potential misuse (\cref{sec:attacks}) and future inconsistencies within the chain. We stress however, that this is not a technical limitation of our solution, but rather a mechanism to remove the burden of the user on deciding what redactions could cause inconsistencies on the chain in the future. We feel that the aforementioned policy is suitable for Bitcoin, but as policies are highly dependent on the application, a different policy can be better suited for different settings.  
%We argue that in a monetary setting such as Bitcoin, the redactions performed are strictly on the non-spendable parts of the transaction, which are most often the non-monetary information in the transaction. This translates to editing only data entries in a transaction. The reason for this restriction is that, permitting any edit operations on monetary information in the transaction could lead to potential misuse (as discussed in \emph{double spend attacks} in~\cref{sec:attacks}) and several future inconsistencies as discussed in the paragraph on \emph{Transaction consistency} below. This restriction can be enforced in the policy $\policy$ of the chain. For instance, the policy could say that only data entries and no monetary information like the addresses, values or other transaction specific information can be edited. Such a policy can be verified for compliance in a chain as discussed in the paragraph on \emph{Accountability} below. 

\paragraph{New Block Structure}  
To account for redactions, the block header must accommodate an additional field called $\mathsf{old\_merkle\_root}$. When a block is initially created, i.e., prior to any redaction, this new field takes the same value as $\mathsf{merkle\_root}$. For a redaction request on block $B_j$, that proposes to replace $\tx_1$ with the candidate transaction $\tx_1^\star$, the transactions list of the candidate block $\candidateblk_j$ (that will replace $B_j$) must contain $\tx_{1_\id} = H(\tx_1)$ in addition to the remaining transactions. A new $\mathsf{merkle\_root}$ is computed for the new set of transactions, while $\mathsf{old\_merkle\_root}$ remains unchanged. To draw parallels with the abstraction we described in~\cref{sec:abstract-new}, $G(s,x)$ is analogous to $\mathsf{merkle\_root}$ and $y$ is analogous to $\mathsf{old\_merkle\_root}$. 

\begin{figure}[h]
\centering
\begin{tabular}{p{.31\columnwidth} p{.55\columnwidth}}
\textbf{Value} & \textbf{Description} \\ \thickhline
$\mathtt{hash\_prev}$ & hash of the previous block header \\ \hline
$\mathtt{merkle\_root}$ & root of the merkle tree (whose the leaves are the transactions)\\ \hline
$\mathtt{difficulty}$ & the difficulty of the proof-of-work\\ \hline
$\mathtt{timestamp}$ & the timestamp of the block \\ \hline
$\mathtt{nonce}$ & nonce used in proof-of-work \\ \thickhline
\rowcolor{Gainsboro!30}$\mathtt{old\_merkle\_root}$ & root of the merkle tree of old set of transactions\\ \thickhline

\end{tabular}%}
\caption{Structure of the Bitcoin block header. The last highlighted field ($\mathtt{old\_merke\_root}$) is only included in the block header of the extended (editable) protocol.}
\label{fig:modifyheader}
\end{figure}

\paragraph{Block Validation} The validation of a block consists of the steps described below.
\begin{itemize}
	\item \emph{Validating transactions}: The block validates all the transactions contained in its transactions list; the validation of non-redacted transactions is performed in the same way as in the immutable version of the protocol. Transactions that have been previously redacted require a special validation that we describe next. Consider the case presented in~\cref{fig:data_insert_removal}, where $\tx_{1}$ is replaced by $\tx_{1}^\star$. The witness $x$ was generated with respect to $\tx_{1_\id}$ and is not valid with respect to $\tx_{1_\id}^\star$. Fortunately, the old state $\tx_{1_\id}$ (hash of the redacted transaction) is stored, as shown in~\cref{fig:sub:after_edit}, ensuring that the witness $x$ can be successfully validated with respect to the old version of the transaction. Therefore, we can ensure that all the transactions included in the block have a valid witness, or in case of redacted transactions, the old version of the transaction had a valid witness. To verify that the redaction was approved in the chain one needs to find a corresponding $\redactTx$ (\cref{fig:editrequest}) in the chain, and verify that it satisfies the chain's policy.   

\item \emph{PoW verification}: The procedure to verify the PoW puzzle is described in~\cref{alg:validateblk}. If the block contains an edited transaction, i.e., $\mathsf{old\_merkle\_root} \ne \mathsf{merkle\_root}$, then substitute the value in $\mathsf{hash\_merkle\_root}$ with that in $\mathsf{old\_merkle\_root}$ and check if the hash of this new header is within $T$. 
  \end{itemize}% If both the above checks are successful, then the block is a \emph{valid} block.
  
%%%%%%%% Modified Tx list
%%%%%%%%%%%%%%%%%%%%%%%%%%%%%%%%%%%%%%%
\begin{figure}[h]
\begin{minipage}[t]{0.45\linewidth}
\centering
\begin{tikzpicture}

   \node [box=4]
    (BE)
    {
    \nodepart[align=center]{one}
    {$\tx_1$}
    \nodepart[align=center]{two}
    {$\tx_2$}
    \nodepart[align=center]{three}
    {$\tx_3$}
        \nodepart[align=center]{four}
    {
    $\rvdots$ 
    }
    };
    \end{tikzpicture}
\subcaption{Non-redacted.}
\label{fig:sub:before_edit}
\end{minipage}
\hfill
\begin{minipage}[t]{0.45\linewidth}

	\centering
	\begin{tikzpicture}
  \node [box=4]
    (BE)
    {
    \nodepart[align=center]{one}
    {$\tx_1^\star, \tx_{1_\id}$}
    \nodepart[align=center]{two}
    {$\tx_2$}
    \nodepart[align=center]{three}
    {$\tx_3$}
        \nodepart[align=center]{four}
    {
    $\rvdots$ 
    }
    };
    \end{tikzpicture}
		\subcaption{Redacted transaction $\tx_1$.}\label{fig:sub:after_edit}
\end{minipage}
\caption{List of transactions contained within a block before (left) and after (right) redacting a transaction in the block.}
\label{fig:modifytxlist}
\end{figure}
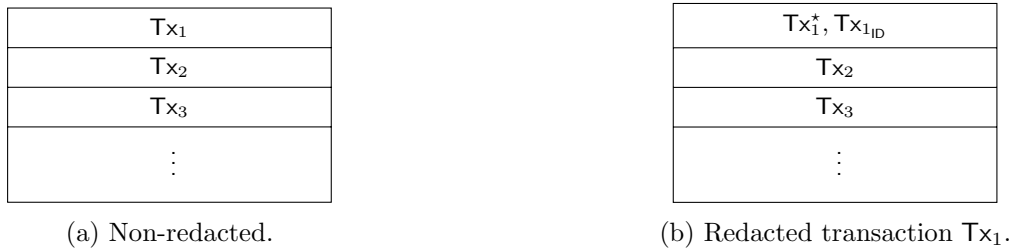

 \paragraph{Chain Validation} To validate a full chain a miner needs to validate all the blocks within the chain. The miner can detect if a block has been redacted by verifying its hash link with the next block; in case of a redacted block, the miner verifies if the redaction was approved according to the chain's policy. The miner rejects a chain as invalid if \emph{any} of the following holds: (1) a block's redaction was not approved according to the policy, (2) the $\mathsf{merkle\_root}$ value of the redacted block is incorrect with respect to the set of transactions (that contains the hash of the redacted transaction) or (3) a previously approved redaction was not performed on the chain.
 %\TODOD{really? I mean, if it is not performed, he can simply perform it himself?}
  
 \paragraph{Transaction Consistency} Removing a transaction entirely or changing spendable data of a transaction may result in serious inconsistencies in the chain. For example, consider a transaction $\tx_1$ that has two outputs denoted by $A$ and $B$, where the second output $B$ has a data entry and the first output $A$ contains a valid spendable script that will be eventually spent by some other transaction $\tx'$. If the redaction operation performed on $\tx_1$ affects the output script of $A$, $\tx'$ may become invalid, causing other transactions to become invalid. A similar problem may arise if the redaction is performed on the input part of $\tx_1$ enabling the user who generated $\tx_1$ to possibly double spend the funds. Therefore, we only allow redactions that do not affect a transaction's consistency with past and future events.
 
%\MESSAGEA{Check if the below paragraph is needed}

  \paragraph{Redaction and Retrievability} The redaction policy $\policy$ for Bitcoin restricts redactions to only those operations that do not violate a transaction's consistency. This means that we do not allow monetary transactions to be edited (such as standard coin transfer). We stress, however that the main objective of redacting a transaction $\tx$ is to prevent some malicious content $x$, that is stored inside $\tx$, from being broadcasted as part of the chain, thereby ensuring that the chain and its users are legally compliant. Note that we cannot prevent an adversary from locally storing and retrieving the data $x$, even after its redaction, since the content was publicly stored in the blockchain. In this case, the user that willingly keeps the malicious (and potentially illegal) data $x$ will be liable.

%  We do this by leveraging the pre-image resistance property of the SHA256 hash function already used in Bitcoin. We store the hash of the original state of the transaction and remove the data entry in the transaction containing malicious content. Therefore, now a miner broadcasts a chain that is devoid of any malicious content. Note that, we obviously cannot prevent a malicious user from storing and retrieving the malicious content locally even after its redaction. In case of existence of such a malicious user, it is this user who is violating the law and not the chain or its users. 
  
 \paragraph{Accountability} Our proposal offers accountability during and after the voting phase is over. Moreover, the accountability during the voting phase prevents the problem of transaction inconsistencies discussed above.
 \begin{itemize}
 	\item \emph{Voting Phase Accountability}:  During the voting phase, anyone can verify all the details of a redaction request. The old transaction and the proposed modification (via the candidate transaction) are up for public scrutiny. It is publicly observable if a miner misbehaves by voting for a redaction request that, apart from removing data, also tampers with the input or (a spendable) output of the transaction, in turn affecting its transaction consistency. This could discourage users from using the system due to its unreliability as a public ledger for monetary purposes. Since the miners are heavily invested in the system and are expected to behave rationally, they would not vote for such an edit request (that is against the policy) during the voting phase.
 	
 	\item \emph{Victim Accountability}: After a redaction is performed, our protocol allows the data owner, whose data was removed, to claim that it was indeed her data that was removed. Since we store the hash of the old transaction along with the candidate transaction in the edited block (refer to~\cref{fig:sub:after_edit}), it is possible for a user that possesses the old data (that was removed) to verify it against the hash that is stored in the redacted block. This enforces accountability on the miners of the network who vote for a redaction request by discouraging them from removing benign data. At the same time, our protocol  guarantees protection against false claims, as the hash verification would fail.
 \end{itemize}

%% file: impl.tex
\section{Proof-of-Concept Implementation}\label{sec:implementation}
In this section we report on a Python proof-of-concept implementation used for evaluating our approach. 
 We implement a full-fledged Blockchain system based on Python 3 that mimics all the basic functionalities of Bitcoin. Specifically, we include a subset of Bitcoin's script language that allows us to insert arbitrary data into the chain, which can be redacted afterwards. The redacting mechanism is built upon the proposed modifications to Bitcoin that we describe in~\cref{sec:instantiation}. For conceptual simplicity we rely on PoW as the consensus mechanism. 

\floatstyle{plain} 
\restylefloat{figure}

\begin{figure}[t]
    \begin{subfigure}[t]{.3\textwidth}
        \centering
        \input{graph1.tex}        
        \label{fig:benchmark_comp}
    \end{subfigure}\hfill
    \begin{subfigure}[t]{.3\textwidth}
        \centering
        \input{graph2.tex}
        \label{fig:benchmark_redact}
    \end{subfigure}\hfill
    \begin{subfigure}[t]{.3\textwidth}
        \centering
        \input{graph3.tex}      
        \label{fig:benchmark_period}
    \end{subfigure}
\caption{The graphs above show the overhead (in percentage) of our proof-of-concept implementation in each experiment performed. In (a) the graph shows the validation time overhead required to validate a redactable chain (with no redactions) compared to an immutable chain; in (b) the graph shows the validation time overhead required to validate a chain for an increasing number of redactions, compared to a redactable chain with no redactions, and finally in (c) the graph shows the validation time overhead required to validate a chain (with $1\%$ of the blocks redacted) for increasing voting periods, compared to a chain (with $1\%$ of the blocks redacted) on a fixed voting period of $\ell = 5$.}
\end{figure}
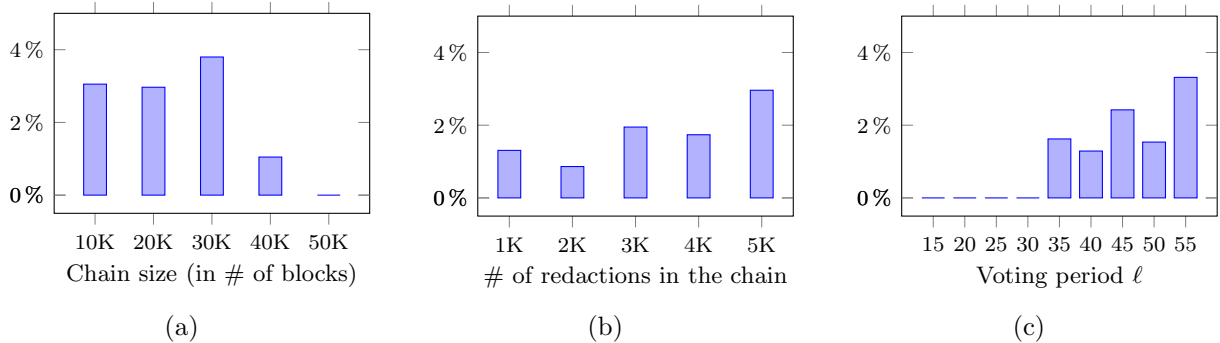

\subsection{Benchmarking}\label{sec:benchmark}
We detail the performance achieved by our implementation running several experiments. The benchmarking was performed in a virtual environment on a Linux server with the following specifications.
\begin{itemize}
        \item Intel Xeon Gold 6132 CPU @ $2.60$GHz
        \item 128GB of RAM
        \item Debian Linux 4.9.0-6-amd64
        \item Python $3.5.3$.
\end{itemize}

We measure the run time of~\cref{alg:validate} by validating chains of varying lengths (i.e., number of blocks) and with different numbers of redactions in the chain. For each experiment, a new chain is created and validated $50$ times, then the arithmetic mean of the run time is taken over all runs.
Each chain consists of up to $50,000$ blocks, where each block contains $1,000$ transactions. Note that a chain of size $50,000$ blocks approximates a one year snapshot of the bitcoin blockchain.

 The great variation of the results shown in the experiments is due to the  randomness involved in the chain creation and validation process, since each chain will contain its own set of (different) transactions, slightly influencing the run time.

 \paragraph{Overhead Compared to Immutable Blockchain} For the first series of experiments, we generate chains of length ranging from $10,000$ up to $50,000$ blocks. We generate both, immutable and redactable chains (with no redactions). 
  The goal here is to measure the overhead that comes with the integration of our redactable blockchain protocol with an immutable blockchain when there are no redactions performed. The results in~\cref{fig:benchmark_comp} indicate that there is only a tiny overhead. Interestingly, we note that as the size of the chain grows, the overhead tends to get smaller; this is because on a chain without redactions the only extra step required is to check if there are any votes in the coinbase transaction of a new block, what becomes negligible compared to the verification time as the chain grows larger.

 \paragraph{Overhead by Number of Redactions} For the second series of experiments, we generate redactable chains with the number of redactions ranging from $2\%$ ($1,000$ redactions) to $10\%$ ($5,000$ redactions) of the blocks. The redacted transactions within a block contains dummy data consisting of $4$ bytes that are removed during the experiment. 
%	Having $2\%$ of the blocks redacted means that we redact one transaction per block and thus $1,000$ transactions considering $50,000$ blocks. 
This experiment is intended to measure the overhead with respect to the number of redactions performed in a chain compared to a redactable chain \emph{with no} redactions. The results in~\cref{fig:benchmark_redact} show that the overhead tends to be at most linear in the number of redactions, since in our prototype instead of looking ahead whether there is a redaction request and a sufficient number of votes, we keep track of the redaction request and wait for its votes and eventual confirmation. 

\paragraph{Overhead by the Voting Parameters $\ell$ and $\rho$} In the last series of experiments, we consider chains with $1\%$ of the blocks redacted. 
We vary the voting period $\ell$ to measure how it influences the validation time compared to a chain with $1\%$ of blocks redacted but with a voting period of $\ell=5$. The threshold of votes $\rho$ is set to $\left(\lfloor\frac{\ell}{2}\rfloor + 1\right)/\ell$ (i.e., requiring majority number of blocks in the voting period to contain votes for approving a redaction). 
	The results in~\cref{fig:benchmark_period} show that the overhead is very small (even negligible for small sizes of $\ell$) and tends to be at most linear in $\ell$. This meets our expectations, since the overhead in validation time originates from keeping and increasing the voting counts over the voting period $\ell$. In the worst case, where $\rho=1$ we need to keep track of the voting count over the entire voting period.% Thus, the worst case overhead is linear in $\ell$.  

%% file: graph1.tex
\begin{tikzpicture}                                                            
  \begin{axis}[
  width=1.2\textwidth,
    height=11em,
    ymin=80,  
    bar width=3mm,
    xtick=data,
    xmin=0.3,
    xmax=5.7,
    ymax=5,
    ymin=-0.5,
    ybar=0,
    extra y ticks       = 0,
    xlabel={\footnotesize Chain size (in \# of blocks)},
    yticklabel=\scriptsize \pgfmathparse{1*\tick}\pgfmathprintnumber{\pgfmathresult}\,\%,
    xticklabel=\scriptsize \pgfmathparse{1*\tick}\pgfmathprintnumber{\pgfmathresult}0K,
    yticklabel style={/pgf/number format/.cd,fixed,precision=2}
        ]  \addplot
coordinates {
(1, 3.053083427557031)
(2, 2.9673926381325466)
(3, 3.7991830096397026)
(4, 1.0473486742365898)
(5, 0)};
%\draw[thin] (axis cs:\pgfkeysvalueof{/pgfplots/xmin},0) -- (axis cs:\pgfkeysvalueof{/pgfplots/xmax},0);
  \end{axis}
\end{tikzpicture}

\caption{}

%% file: graph2.tex
\begin{tikzpicture}
  \begin{axis}[
  width=1.2\textwidth,
    height=11em,
    ymin=80,  
    bar width=3mm,
    xtick=data,
    xmin=0.5,
    xmax=5.5,
    ymax=5,
    ymin=-0.5,
    ybar=0,
    extra y ticks       = 0,
    xlabel={\footnotesize $\#$ of redactions in the chain},
    yticklabel=\scriptsize \pgfmathparse{1*\tick}\pgfmathprintnumber{\pgfmathresult}\,\%,
    xticklabel=\scriptsize \pgfmathparse{1*\tick}\pgfmathprintnumber{\pgfmathresult}K,
    yticklabel style={/pgf/number format/.cd,fixed,precision=2}
        ]
  \addplot
coordinates {
(1, 1.3070767719031349)
(2, 0.8628298316230483)
(3, 1.9483735313250665)
(4, 1.7370226878977915)
(5, 2.9612391807180543)
};
%\draw[thin] (axis cs:\pgfkeysvalueof{/pgfplots/xmin},0) -- (axis cs:\pgfkeysvalueof{/pgfplots/xmax},0);
  \end{axis}
\end{tikzpicture}

\caption{}

%% file: graph3.tex
\begin{tikzpicture}
  \begin{axis}[
  width=1.2\textwidth,
    height=11em,
    ymin=80,  
    bar width=3mm,
    xtick=data,
    xmin=10,
    xmax=60,
    ymax=5,
    ymin=-0.5,
    ybar=0,
    extra y ticks       = 0,
    xlabel={\footnotesize Voting period $\ell$},
    xticklabel=\scriptsize \pgfmathparse{1*\tick}\pgfmathprintnumber{\pgfmathresult},
    yticklabel=\scriptsize \pgfmathparse{1*\tick}\pgfmathprintnumber{\pgfmathresult}\,\%,
    yticklabel style={/pgf/number format/.cd,fixed,precision=2}
        ]
  \addplot
coordinates {
%(10, 0.8395723484928953)
(15, 0)
(20, 0)
(25, 0)
(30,0)
(35,1.6217981388567646)
(40,1.2889899688943318)
(45,2.4215612763945873)
(50,1.5354370027936013)
(55,3.3150705290616647)
%(60,0)
};
%\draw[thin] (axis cs:\pgfkeysvalueof{/pgfplots/xmin},0) -- (axis cs:\pgfkeysvalueof{/pgfplots/xmax},0);
  \end{axis}
  
  \end{tikzpicture}

\caption{}

%% file: attacks.tex
\section{Discussion}\label{sec:attacks}
In this section we discuss some of the generic attacks on our system and how it is immune to such attacks. 

%\TODOB{We should be consistent whether we say ``user" or ``miner", or at least say that we use the term interchangeably since we mean the same.}

\paragraph{Unapproved Editing} A malicious miner could pass off an edit on the blockchain that does not satisfy the network's policy. This can occur if the miner presents the blockchain with an edit that has not been considered for voting, or has gathered insufficient votes. In any of the above cases, it is possible for any user in the network to account for an edit by verifying in the chain if the exact edit presented by the miner is approved or not. And since majority of the miners in the network is honest, the user accepts an approved edit as an honest edit.

%\TODOB{Rewrite the following paragraph.}
\paragraph{Scrutiny of Candidate Blocks} It is in the interest of the (honest) miners and the system as a whole, to actively scrutinise a candidate block and decide on voting based on its merit.
Therefore, the miners are strongly discouraged from using a default strategy in voting, e.g., always vote for a candidate block without scrutiny, using a pre-determined strategy that is agnostic to what the candidate block is proposing. %\paragraph{Scrutiny of Candidate Blocks} Consider a situation where a user proposes a candidate block to the network. Given accountability during the voting phase, the miners in the network, while deliberating whether to vote for the candidate block, must behave rationally. This strongly discourages a \emph{lazy} miner from using a default strategy in voting, for instance, always vote yes or always vote no for any edit request blindly without scrutiny. In other words, miners are discouraged from voting with a pre-determined strategy that is agnostic to what an edit request is proposing. It is therefore in the interest of the miners and the system as a whole, to actively scrutinise an edit request and decide on voting based on the edit request's merit.

\paragraph{Denial of Service} A malicious miner may try to flood the network with edit requests as an attempt to slow down transaction confirmation in the chain. However, the miner is deterred from doing this because he incurs the cost of a transaction fee for the $\redactTx$ that is part of his edit request similar to other standard transactions. Moreover, it may also be the case for the $\redactTx$ to incur a higher transaction fee as a strong deterrent against spamming. 

\paragraph{False Victim} A malicious user may wrongly claim that a particular transaction related to him was edited. For example, he may claim that some monetary information was changed where he was the beneficiary. Since such an edit could affect the trust in the system, the user could potentially affect the credibility of the entire system. We prevent such an attack through  victim accountability of our protocol. We can verify the user's claim against the hash of the old version of the transaction that is stored in the chain itself. Given the hash function is collision resistant, a wrong claim would fail the check. 

\paragraph{Double Spend Attacks} Consider a scenario where a malicious user is the recipient of a transaction. If this transaction was edited by removing some data stored in it, the hash of the new version of the transaction is different. If the miner had already spent the funds from the old version of this transaction, after the edit, he may attempt a double spend by exploiting the new version of the transaction. This is prevented by associating the new version and the old version of the edited transaction with each other, thereby noticing such a double spend. If the funds had already been spent, the old version would be a spent transaction. Because the edit that is performed does not conflict with the consistency of the transaction, the new version of the transaction would also be a spent transaction.

\paragraph{Consensus delays} Consider a scenario where two different users hold chains with a different set of redacted blocks, and therefore cannot arrive at a consensus on the final state of the chain, what may result in delays. Assuming the miners have not locally redacted blocks on their own and have behaved honestly according to the protocol, this scenario would mean that the different set of redacted blocks in the chains held by the two miners have been approved by the policy. However, this would be a blatant violation of the Editable common prefix property of our protocol (\Cref{thm:editable_common_prefix}).

  \section{Related work}
%  \TODOB{We leave this here for now}
%\TODOD{Blockchain and consensus have already been discussed in the intro?!}
% \subsubsection{Blockchain and Consensus} Nakamoto proposed Bitcoin~\cite{nakamoto2008bitcoin}, the first cryptocurrency system with a consensus protocol based on Proof of Work (PoW), which was originally introduced by Dwork and Naor~\cite{C:DwoNao92}. The underlying protocol of Bitcoin was dubbed as the \emph{Blockchain}, and a formal analysis of its security definitions and properties were shown in the works of Garay et al.~\cite{EC:GarKiaLeo15,garay2017bitcoin} and Pass et al.~\cite{EC:PassSS17}. Many alternate consensus mechanisms were proposed and formalised, such as Proof of Stake (PoS)~\cite{kiayias2017ouroboros,david2018ouroboros,bentov2016snow,cryptoeprint:2018:378}, Proof of Space~\cite{C:DFKP15}, and Proof of Secure Erasure~\cite{SCN:ABFG14,SCN:KarKia14}.  
  
   \ifieeesp\subsubsection{Bitcoin and Applications} \else \paragraph{Bitcoin and Applications} \fi Several works~\cite{andrychowicz2014modeling,ateniese2014certified,CCS:PasShe15}  have analysed the properties and extended the features of the Bitcoin protocol. Bitcoin as a public bulletin board has found several innovative applications far beyond its initial scope, e.g., to achieve fairness and correctness in secure multi-party computation~\cite{SP:ADMM14,FCW:ADMM14,C:BenKum14,CCS:KumBen14}, to build smart contracts~\cite{SP:KMSWP16}, to distributed cryptography~\cite{C:AndDzi15}, and more~\cite{CCS:KumMorBen15,CCS:KiaTan15,breidenbach2018enter}.
   
   \ifieeesp\subsubsection{Content Insertion in Bitcoin}\else \paragraph{Content Insertion in Bitcoin}\fi There have been several works~\cite{ateniese2017redactable,matzutt2016poster,mcreynolds2015cryptographic,puddu2017muchain,shirriff2014hidden,sleiman2015bitcoin} on analysing and assessing the consequences of content insertions in public blockchains. They shed light on the distribution and the usage of such inserted data entries. The most recent work of Matzutt et al.~\cite{matzutt2018quantitative} gives a comprehensive quantitative analysis of illicit content insertions including, insertion techniques, potential risks and rational incentives. They also show that compared to other attacks~\cite{FC:EyaSir14,heilman2015eclipse} on Bitcoin system, illicit content insertion can pose immediate risks to all users of the system.
   
   \ifieeesp\subsubsection{Proactive Countermeasures} \else \paragraph{Proactive Countermeasures} \fi Proactive measures to detect illicit material circulated in the network and detecting them have been studied~\cite{roesch1999snort,CCS:IKBS00,hu2014flowguard}. In a blockchain setting, preventive solutions~\cite{chepurnoy2016rollerchain,bruce2014mini,molina2017pascalcoin} focus on maintaining only monetary information instead of the entire ledger history. Matzutt et al.~\cite{matzutt2018thwarting} use a rational approach of discouraging miners from inserting harmful content into the blockchain. They advocate a minimum transaction fee and mitigation of transaction manipulatability as a deterrent for the same.

%% file: appendix.tex
\appendix

\ifieeesp\subsection{Protocol extension for multiple redactions} \label{apx:extension}\else \section{Protocol extension for multiple redactions} \label{apx:extension} \fi

In this section we sketch an extension to the protocol of~\cref{fig:protocol} to accommodate multiple redactions per block.

The intuition behind the extension is simple enough to be explained in this paragraph; a block can potentially be redacted $n$ times and each redaction $\candidateblk_j$ of the block $B_j$ that is approved \emph{must} contain information about the entire history of previous redactions. In our extension, this information is stored in the $y^\star_j$ component of the candidate block $\candidateblk_j$. We now sketch the required protocol changes. 

\paragraph{Proposing an Edit} To propose a redaction for block $B_j := \langle s_j,x_j,\ctr_j, y_j\rangle$ the user must build a candidate block $\candidateblk_j$ of the following form: $\candidateblk_j := \langle s_j,x^\star_j,\ctr_j, y^\star_j \rangle$, where $y^\star_j := y_j || G(s_j,x_j)$ iff $y_j \ne G(s_j,x_j)$. Note that for the first redaction of $B_j$, we have that $y_j = G(s_j,x_j)$, and therefore $y^\star_j:=G(s_j,x_j)$.

\paragraph{Block Validation} To validate a block, the users run the $\blockValidext$ algorithm described in~\cref{alg:validateblk2}. Intuitively, the algorithm performs the same operations as~\cref{alg:validateblk}, except that it takes into account the possibility of the block being redacted multiple times. Observe that by parsing $y$ as $y^{(1)}||y^{(2)}||...||y^{(l)}$, we are considering a block that has been redacted a total of $l$ times and $y^{(1)}$ denotes the original state information of the unredacted version of the block.

\begin{algorithm}
\footnotesize
\SetAlgoVlined
\SetKwInOut{Input}{input}\SetKwInOut{Output}{output}
\Input{Block $B:=\blockmod$.}
\Output{$\bit$}
\BlankLine
Validate data $x$, if \emph{invalid} \Return $0$\;
Parse $y$ as $y^{(1)}||y^{(2)}||...||y^{(l)}$, where $y_j^{(i)} \in \bit^\spar$ $\forall i\in[l]$\;
\lIf{$(H(\ctr,G(s,x),y) < D) \lor (H(\ctr, y^{(1)},y^{(1)}) < D)$\\}{\Return $1$} \label{line:powext}
\lElse{\Return $0$}
\caption{$\blockValidext$}	
\label{alg:validateblk2}
\end{algorithm} 

\begin{algorithm}
\footnotesize
	\SetKwInOut{Input}{input}\SetKwInOut{Output}{output}
\Input{Chain $\chain = (B_1,\cdots,B_n)$ of length $n$, and a  candidate block $\candidateblk_j$ for an edit.}
\Output{$\bit$}
\BlankLine
Parse $\candidateblk_j := \langle s_j,x^\star_j, \ctr_j, y_j \rangle$\;
Parse $y_j$ as $y_j^{(1)}||y_j^{(2)}||...||y_j^{(l)}$, where $y_j^{(i)} \in \bit^\spar$ $\forall i \in [l]$\;
\lIf{$\Blk'.\blockValidext(\candidateblk_j) = 0$}{\Return $0$}
Parse $B_{j-1}:= \blockmodi{j-1}$\;
Parse $y_{j-1}$ as $y_{j-1}^{(1)}||y_{j-1}^{(2)}||...||y_{j-1}^{(l')}$, where $y_{j-1}^{(i)} \in \bit^\spar$ $\forall i \in [l']$\;
Parse $B_{j+1}:= \blockmodi{j+1}$\;
\lIf{$s_j \ne H(\ctr_{j-1},y_{j-1}^{(1)},y_{j-1}^{(1)}) \lor s_{j+1} \ne H(\ctr_j,y^{(1)}_j,y^{(1)}_j)$}{\Return $0$}
\For{$i \in \{2,\ldots,n\}$} {\lIf{ the fraction of votes for $H(\ctr,y_j^{(i)},y_j^{(1)}||\ldots||y_j^{(i-1)})$ in the chain $\chain$ is \emph{not} at least $\rho$ within its voting period of $\ell$ blocks}{\Return $0$}}
\Return $1$\label{line:links2}

\caption{$\decideext$}
\label{alg:candidate2}
\end{algorithm}

\begin{algorithm}
\footnotesize
\SetAlgoVlined
\SetKwInOut{Input}{input}\SetKwInOut{Output}{output}
\Input{Chain $\chain = (B_1,\cdots,B_n)$ of length $n$.}
\Output{$\bit$}
\BlankLine
$j:=n$\;
\lIf{$j = 1$}{\Return $\Blk'.\blockValidext(B_1)$}
\While{$j \ge 2$}{ 
$B_j:= \langle s_j,x_j,\ctr_j,y_j\rangle$ \Comment*{\scriptsize{$B_j := \head{\chain}$ when $j=n$}}
$B_{j-1:}:= \langle s_{j-1},x_{j-1},\ctr_{j-1},y_{j-1}\rangle$\;
Parse $y_j$ as $y_j^{(1)}||\dots||y_j^{(l)}$, where $y_j^{(i)} \in \bit^\spar$ $\forall i\in[l]$\;
Parse $y_{j-1}$ as $y_{j-1}^{(1)}||\dots||y_{j-1}^{(l')}$, where $y_{j-1}^{(i)} \in \bit^\spar$ $\forall i\in[l']$\;
\lIf{$\Blk'.\blockValidext(B_j) = 0$}{\Return $0$} 
\lIf{$s_j = H(\ctr_{j-1}, G(s_{j-1},x_{j-1}),y_{j-1})$}{$j := j-1$}\label{line:new_state2}
\lElseIf{$s_j = H(\ctr_{j-1}, y^{(1)}_{j-1},y^{(1)}_{j-1}) \land \Blk'.\decideext(\chain,B_{j-1}) = 1 \land \policy(\chain,B_{j-1}) =\accept$}{$j := j-1$}\label{line:old_state2}
\lElse{\Return $0$}
}
\Return $1$\;
\caption{$\chainValidext$}	
\label{alg:validate2}
\end{algorithm}

\paragraph{Voting for Candidate Blocks} To vote for a redaction, we additionally define the following interface. 
\begin{itemize}
	\item $H(\ctr,G(s,x^\star),y^\star) \gets \Blk'.\votefunc(\candidateblk)$: takes as input a candidate block $\candidateblk$ and parses $\candidateblk$ as $(s,x^\star,\ctr,y^\star)$. It outputs the hash value $H(\ctr,G(s,x^\star),y^\star)$ as a vote for the candidate block $\candidateblk$.
	
\end{itemize}

The voting interface is invoked by users that wish to endorse a candidate block by including a vote in the newly mined block (if the candidate block is still in its voting phase).
Accordingly the policy $\policy$ of the chain for redactions checks if a candidate block has received at least a ratio of $\rho$ votes (as output by the $\Blk'.\votefunc$) in a span of $\ell$ blocks after immediately its proposal.

\paragraph{Candidate Block Validation} If a block $B_j$ is being redacted more than once, then the corresponding candidate block $B^\star_j$ needs to be validated for accounting for the multiple redactions that happened before; for each redaction of $B_j$, the votes for that redaction must exist in the chain $\chain$. $\decideext$ (described in~\cref{alg:candidate2}) validates such a candidate block.

\paragraph{Chain Validation} To validate a chain, the user runs the $\chainValidext$ algorithm (described in~\cref{alg:validate2}). The only change compared to~\cref{alg:validate} is that now $y_j$ is parsed as $y_j^{(1)}||\dots||y_j^{(l)}$ where the initial unredacted state of the block is stored in $y^{(1)}$.

%\MESSAGEB{Check consistency of the components with and without $\star$ and don't forget to define and accommodate for the new vote function in the appendix.}

%% file: ms.bbl
\begin{thebibliography}{10}

\bibitem{script}
Bitcoin script.
\newblock \url{https://en.bitcoin.it/wiki/Script}.

\bibitem{interpol2015}
Interpol cyber research identifies malware threat to virtual currencies.
\newblock {\em Interpol}, 2015.
\newblock \url{https://tinyurl.com/y9wfekr6}.

\bibitem{businessinsider}
Accenture files patent for editable blockchain.
\newblock {\em Business Insider Deutschland}, 2016.
\newblock \url{https://tinyurl.com/yblq9zdp}.

\bibitem{nytimes}
Downside of bitcoin: A ledger that can’t be corrected.
\newblock {\em The New York Times}, 2016.
\newblock \url{https://tinyurl.com/ydxjlf9e}.

\bibitem{coindesk2018}
Child porn on bitcoin? why this doesn't mean what you might think.
\newblock {\em Coindesk}, 2018.
\newblock \url{https://tinyurl.com/y8yo8ml3}.

\bibitem{C:AndDzi15}
Marcin Andrychowicz and Stefan Dziembowski.
\newblock {PoW}-based distributed cryptography with no trusted setup.
\newblock In Rosario Gennaro and Matthew J.~B. Robshaw, editors, {\em Advances
  in Cryptology -- {CRYPTO}~2015, Part II}, volume 9216 of {\em Lecture Notes
  in Computer Science}, pages 379--399. Springer, Heidelberg, August 2015.

\bibitem{FCW:ADMM14}
Marcin Andrychowicz, Stefan Dziembowski, Daniel Malinowski, and Lukasz Mazurek.
\newblock Fair two-party computations via bitcoin deposits.
\newblock In Rainer B{\"o}hme, Michael Brenner, Tyler Moore, and Matthew Smith,
  editors, {\em FC 2014 Workshops}, volume 8438 of {\em Lecture Notes in
  Computer Science}, pages 105--121. Springer, Heidelberg, March 2014.

\bibitem{andrychowicz2014modeling}
Marcin Andrychowicz, Stefan Dziembowski, Daniel Malinowski, and {\L}ukasz
  Mazurek.
\newblock Modeling bitcoin contracts by timed automata.
\newblock In {\em International Conference on Formal Modeling and Analysis of
  Timed Systems}, pages 7--22. Springer, 2014.

\bibitem{SP:ADMM14}
Marcin Andrychowicz, Stefan Dziembowski, Daniel Malinowski, and Lukasz Mazurek.
\newblock Secure multiparty computations on bitcoin.
\newblock In {\em 2014 {IEEE} Symposium on Security and Privacy}, pages
  443--458. {IEEE} Computer Society Press, May 2014.

\bibitem{ateniese2017multiple}
Giuseppe Ateniese, Michael~T Chiaramonte, David Treat, Bernardo Magri, and
  Daniele Venturi.
\newblock Multiple-link blockchain, October~10 2017.
\newblock US Patent 9,785,369.

\bibitem{ateniese2018rewritable}
Giuseppe Ateniese, Michael~T Chiaramonte, David Treat, Bernardo Magri, and
  Daniele Venturi.
\newblock Rewritable blockchain, May~8 2018.
\newblock US Patent 9,967,096.

\bibitem{ateniese2014certified}
Giuseppe Ateniese, Antonio Faonio, Bernardo Magri, and Breno De~Medeiros.
\newblock Certified bitcoins.
\newblock In {\em International Conference on Applied Cryptography and Network
  Security}, pages 80--96. Springer, 2014.

\bibitem{ateniese2017redactable}
Giuseppe Ateniese, Bernardo Magri, Daniele Venturi, and Ewerton Andrade.
\newblock Redactable blockchain--or--rewriting history in bitcoin and friends.
\newblock In {\em Security and Privacy (EuroS\&P), 2017 IEEE European Symposium
  on}, pages 111--126. IEEE, 2017.

\bibitem{cryptoeprint:2018:378}
Christian Badertscher, Peter Ga\v{z}i, Aggelos Kiayias, Alexander Russell, and
  Vassilis Zikas.
\newblock Ouroboros genesis: Composable proof-of-stake blockchains with dynamic
  availability.
\newblock In {\em Proceedings of the 2018 ACM SIGSAC Conference on Computer and
  Communications Security}, CCS '18, pages 913--930.

\bibitem{C:BenKum14}
Iddo Bentov and Ranjit Kumaresan.
\newblock How to use bitcoin to design fair protocols.
\newblock In Juan~A. Garay and Rosario Gennaro, editors, {\em Advances in
  Cryptology -- {CRYPTO}~2014, Part II}, volume 8617 of {\em Lecture Notes in
  Computer Science}, pages 421--439. Springer, Heidelberg, August 2014.

\bibitem{breidenbach2018enter}
Lorenz Breidenbach, IC~Cornell~Tech, Philip Daian, Florian Tramer, and Ari
  Juels.
\newblock Enter the hydra: Towards principled bug bounties and
  exploit-resistant smart contracts.
\newblock In {\em 27th USENIX Security Symposium (USENIX Security 18)}. USENIX
  Association, 2018.

\bibitem{bruce2014mini}
JD~Bruce.
\newblock The mini-blockchain scheme.
\newblock {\em White paper}, 2014.

\bibitem{camenisch2017chameleon}
Jan Camenisch, David Derler, Stephan Krenn, Henrich~C P{\"o}hls, Kai Samelin,
  and Daniel Slamanig.
\newblock Chameleon-hashes with ephemeral trapdoors.
\newblock In {\em IACR International Workshop on Public Key Cryptography},
  pages 152--182. Springer, 2017.

\bibitem{chepurnoy2016rollerchain}
Alexander Chepurnoy, Mario Larangeira, and Alexander Ojiganov.
\newblock Rollerchain, a blockchain with safely pruneable full blocks.
\newblock {\em arXiv preprint arXiv:1603.07926}, 2016.

\bibitem{FC:EyaSir14}
Ittay Eyal and Emin~G{\"u}n Sirer.
\newblock Majority is not enough: Bitcoin mining is vulnerable.
\newblock In Nicolas Christin and Reihaneh {Safavi-Naini}, editors, {\em FC
  2014: 18th International Conference on Financial Cryptography and Data
  Security}, volume 8437 of {\em Lecture Notes in Computer Science}, pages
  436--454. Springer, Heidelberg, March 2014.

\bibitem{EC:GarKiaLeo15}
Juan~A. Garay, Aggelos Kiayias, and Nikos Leonardos.
\newblock The bitcoin backbone protocol: Analysis and applications.
\newblock In Elisabeth Oswald and Marc Fischlin, editors, {\em Advances in
  Cryptology -- {EUROCRYPT}~2015, Part II}, volume 9057 of {\em Lecture Notes
  in Computer Science}, pages 281--310. Springer, Heidelberg, April 2015.

\bibitem{gipp2016securing}
Bela Gipp, Jagrut Kosti, and Corinna Breitinger.
\newblock Securing video integrity using decentralized trusted timestamping on
  the bitcoin blockchain.
\newblock In {\em MCIS}, page~51, 2016.

\bibitem{gipp2015decentralized}
Bela Gipp, Norman Meuschke, and Andr{\'e} Gernandt.
\newblock Decentralized trusted timestamping using the crypto currency bitcoin.
\newblock {\em arXiv preprint arXiv:1502.04015}, 2015.

\bibitem{heilman2015eclipse}
Ethan Heilman, Alison Kendler, Aviv Zohar, and Sharon Goldberg.
\newblock Eclipse attacks on bitcoin's peer-to-peer network.
\newblock In {\em USENIX Security Symposium}, pages 129--144, 2015.

\bibitem{hu2014flowguard}
Hongxin Hu, Wonkyu Han, Gail-Joon Ahn, and Ziming Zhao.
\newblock Flowguard: building robust firewalls for software-defined networks.
\newblock In {\em Proceedings of the third workshop on Hot topics in software
  defined networking}, pages 97--102. ACM, 2014.

\bibitem{ibanez2018blockchains}
Luis-Daniel Ibanez, Kieron O'Hara, and Elena Simperl.
\newblock On blockchains and the general data protection regulation.
\newblock 2018.

\bibitem{CCS:IKBS00}
Sotiris Ioannidis, Angelos~D. Keromytis, Steven~M. Bellovin, and Jonathan~M.
  Smith.
\newblock Implementing a distributed firewall.
\newblock In S.~Jajodia and P.~Samarati, editors, {\em ACM CCS 00: 7th
  Conference on Computer and Communications Security}, pages 190--199. {ACM}
  Press, November 2000.

\bibitem{kiayias2017ouroboros}
Aggelos Kiayias, Alexander Russell, Bernardo David, and Roman Oliynykov.
\newblock Ouroboros: A provably secure proof-of-stake blockchain protocol.
\newblock In {\em Annual International Cryptology Conference}, pages 357--388.
  Springer, 2017.

\bibitem{CCS:KiaTan15}
Aggelos Kiayias and Qiang Tang.
\newblock Traitor deterring schemes: Using bitcoin as collateral for digital
  content.
\newblock In Indrajit Ray, Ninghui Li, and Christopher Kruegel:, editors, {\em
  ACM CCS 15: 22nd Conference on Computer and Communications Security}, pages
  231--242. {ACM} Press, October 2015.

\bibitem{SP:KMSWP16}
Ahmed~E. Kosba, Andrew Miller, Elaine Shi, Zikai Wen, and Charalampos
  Papamanthou.
\newblock Hawk: The blockchain model of cryptography and privacy-preserving
  smart contracts.
\newblock In {\em 2016 {IEEE} Symposium on Security and Privacy}, pages
  839--858. {IEEE} Computer Society Press, May 2016.

\bibitem{CCS:KumBen14}
Ranjit Kumaresan and Iddo Bentov.
\newblock How to use bitcoin to incentivize correct computations.
\newblock In Gail-Joon Ahn, Moti Yung, and Ninghui Li, editors, {\em ACM CCS
  14: 21st Conference on Computer and Communications Security}, pages 30--41.
  {ACM} Press, November 2014.

\bibitem{CCS:KumMorBen15}
Ranjit Kumaresan, Tal Moran, and Iddo Bentov.
\newblock How to use bitcoin to play decentralized poker.
\newblock In Indrajit Ray, Ninghui Li, and Christopher Kruegel:, editors, {\em
  ACM CCS 15: 22nd Conference on Computer and Communications Security}, pages
  195--206. {ACM} Press, October 2015.

\bibitem{matzutt2018thwarting}
Roman Matzutt, Martin Henze, Jan~Henrik Ziegeldorf, Jens Hiller, and Klaus
  Wehrle.
\newblock Thwarting unwanted blockchain content insertion.
\newblock In {\em Cloud Engineering (IC2E), 2018 IEEE International Conference
  on}, pages 364--370. IEEE, 2018.

\bibitem{matzutt2018quantitative}
Roman Matzutt, Jens Hiller, Martin Henze, Jan~Henrik Ziegeldorf, Dirk
  M{\"u}llmann, Oliver Hohlfeld, and Klaus Wehrle.
\newblock A quantitative analysis of the impact of arbitrary blockchain content
  on bitcoin.
\newblock In {\em Proceedings of the 22nd International Conference on Financial
  Cryptography and Data Security (FC). Springer}, 2018.

\bibitem{matzutt2016poster}
Roman Matzutt, Oliver Hohlfeld, Martin Henze, Robin Rawiel, Jan~Henrik
  Ziegeldorf, and Klaus Wehrle.
\newblock Poster: I don't want that content! on the risks of exploiting
  bitcoin's blockchain as a content store.
\newblock In {\em Proceedings of the 2016 ACM SIGSAC conference on computer and
  communications security}, pages 1769--1771. ACM, 2016.

\bibitem{mcreynolds2015cryptographic}
Emily McReynolds, Adam Lerner, Will Scott, Franziska Roesner, and Tadayoshi
  Kohno.
\newblock Cryptographic currencies from a tech-policy perspective: Policy
  issues and technical directions.
\newblock In {\em International Conference on Financial Cryptography and Data
  Security}, pages 94--111. Springer, 2015.

\bibitem{molina2017pascalcoin}
Albert Molina and Herman Schoenfeld.
\newblock Pascalcoin version 2.
\newblock {\em White paper}, 2017.

\bibitem{nakamoto2008bitcoin}
Satoshi Nakamoto.
\newblock Bitcoin: A peer-to-peer electronic cash system, 2008.

\bibitem{EC:PassSS17}
Rafael Pass, Lior Seeman, and Abhi Shelat.
\newblock Analysis of the blockchain protocol in asynchronous networks.
\newblock In {\em Advances in Cryptology - {EUROCRYPT} 2017 - 36th Annual
  International Conference on the Theory and Applications of Cryptographic
  Techniques, Paris, France, April 30 - May 4, 2017, Proceedings, Part {II}},
  pages 643--673, 2017.

\bibitem{CCS:PasShe15}
Rafael Pass and Abhi Shelat.
\newblock Micropayments for decentralized currencies.
\newblock In Indrajit Ray, Ninghui Li, and Christopher Kruegel:, editors, {\em
  ACM CCS 15: 22nd Conference on Computer and Communications Security}, pages
  207--218. {ACM} Press, October 2015.

\bibitem{pass2017fruitchains}
Rafael Pass and Elaine Shi.
\newblock Fruitchains: A fair blockchain.
\newblock In {\em Proceedings of the ACM Symposium on Principles of Distributed
  Computing}, pages 315--324. ACM, 2017.

\bibitem{puddu2017muchain}
Ivan Puddu, Alexandra Dmitrienko, and Srdjan Capkun.
\newblock $\mu$chain: How to forget without hard forks.
\newblock {\em IACR Cryptology ePrint Archive}, 2017:106, 2017.

\bibitem{roesch1999snort}
Martin Roesch et~al.
\newblock Snort: Lightweight intrusion detection for networks.
\newblock In {\em Lisa}, volume~99, pages 229--238, 1999.

\bibitem{shirriff2014hidden}
Ken Shirriff.
\newblock Hidden surprises in the bitcoin blockchain and how they are stored:
  Nelson mandela, wikileaks, photos, and python software.
\newblock {\em Ken Shirriff’s blog (accessed July 2017) http://www. righto.
  com/2014/02/ascii-bernanke-wikileaks-photographs. html}, 2014.

\bibitem{sleiman2015bitcoin}
Matthew~D Sleiman, Adrian~P Lauf, and Roman Yampolskiy.
\newblock Bitcoin message: Data insertion on a proof-of-work cryptocurrency
  system.
\newblock In {\em Cyberworlds (CW), 2015 International Conference on}, pages
  332--336. IEEE, 2015.

\bibitem{tziakouris2018cryptocurrencies}
Giannis Tziakouris.
\newblock Cryptocurrencies—a forensic challenge or opportunity for law
  enforcement? an interpol perspective.
\newblock {\em IEEE Security \& Privacy}, 16(4):92--94, 2018.

\end{thebibliography}
